\documentclass{article}
\usepackage{a4wide}
\usepackage{amsmath,amsfonts,amssymb,amsthm}
\newcommand{\R}{{\mathbb R}}

\newcommand{\y}{{\bf y}}
\newcommand{\C}{{\rm const}}
\newcommand{\x}{{\bf x}}

\newcommand{\n}{{\bf n}}

\newtheorem{theorem}{Theorem}[section]

\newtheorem{proposition}[theorem]{Proposition}
\newtheorem{corollary}[theorem]{Corollary}
\newtheorem{lemma}[theorem]{Lemma}

\theoremstyle{definition}

\numberwithin{equation}{section}

\begin{document}

\noindent 
\begin{center}
\textbf{\large The Faraday effect revisited: Thermodynamic limit}
\end{center}

\begin{center}
February 28th, 2009
\end{center}

\vspace{0.5cm}

\noindent 

\begin{center}
\textbf{ 
Horia D. Cornean\footnote{Department of Mathematical Sciences, 
    Aalborg
    University, Fredrik Bajers Vej 7G, 9220 Aalborg, Denmark; e-mail:
    cornean@math.aau.dk},
Gheorghe Nenciu\footnote{Dept. Theor. Phys.,
University of Bucharest, P. O. Box MG11, 
RO-76900 Bu\-cha\-rest, Romania; }${}^,$
\footnote{
Inst. of Math. ``Simion Stoilow'' of
the Romanian Academy, P. O. Box 1-764, RO-014700 Bu\-cha\-rest, Romania;
e-mail: Gheorghe.Nenciu@imar.ro}}
     
\end{center}

\vspace{0.5cm}

\noindent

\begin{abstract}
 This paper is the second in a series revisiting 
   the (effect of) Faraday rotation. We formulate and 
prove the thermodynamic limit for the transverse electric 
conductivity of Bloch electrons, as well as for the Verdet
constant. The main mathematical tool is a regularized magnetic and
geometric perturbation theory combined with elliptic regularity and Agmon-Combes-Thomas uniform
exponential decay estimates.  
\end{abstract}

\vspace{0.5cm}

\tableofcontents

\newpage

\section{Introduction}

\subsection{Generalities}
The rotation of the polarization of a plane-polarized electromagnetic
wave passing through a 
material immersed in a homogeneous magnetic field oriented parallel to the direction of 
propagation, is known in physics as  the Faraday effect (called sometimes also Faraday rotation).
The experiment consists in sending a monochromatic light wave,
parallel to the $0z$ direction and linearly polarized in the plane
$x0z$. When the light enters the sample, the polarization plane starts
rotating. A simple argument based on (classical) Maxwell equations
shows that there exists a linear relation between the angle $\theta$
of rotation of the plane of polarization per unit length and the
transverse component of the conductivity tensor of the material (see
e.g. formula (1) in \cite{Roth}). For most materials - and we will
restrict ourselves to this case -  the transverse component of the conductivity tensor
vanishes when the magnetic field is absent and is no longer zero when
the magnetic field is turned on. Under the proviso that the dependence
of the conductivity tensor upon the strength $B$ of the magnetic
field is smooth, for weak fields one expands the conductivity tensor
to the first order and neglect the higher terms. The coefficient of
the linear term is known as the Verdet constant of the corresponding material.

It follows that the basic object is the conductivity tensor and the
main goal of the theory (classical or quantum) is to provide a
workable formula for it, in particular for the Verdet constant. 
The problem has a long and distinguished history in solid state
physics  theory and the spectrum of possible applications ranges from
astrophysics to optics and general quantum mechanics 
(see e.g. \cite{Roth, Ped, ZP, CNP, MPP, HaJa, Kas} and references therein). 

Using quantum theory in the setting in which the sample is modeled by 
a system of independent electrons subjected to a periodic electric
potential, Laura Roth \cite{Roth} obtained  (albeit only at a formal
level) a formula for the Verdet constant in full generality and
applied it to metals as well as semiconductors. Roth's method is based
on an effective Hamiltonian approach for Bloch electrons in the
presence of a weak constant magnetic field (see \cite{Roth1}
and references therein) which in turn is based on a 
(proto) magnetic pseudodifferential calculus (for recent mathematical
developments see \cite{IMP} and references therein).
 
But Roth's theory is far from being free of difficulties. Due to her
highly formal way of doing computations, it seems almost hopeless -
even with present day mathematical tools - to control the errors or push the
computations to higher orders in $B$ except maybe the case of simple
bands. Even more, the final formula contains terms which are singular
at the crossings of Bloch bands. Accordingly, in spite of the fact
that it has been considered a landmark of the subject, it came as no
surprise that at the practical level this theory  only met a moderate
success and a multitude of unrelated, simplified models have been tailored for specific cases.

Our paper is the second in a series aiming at a complete, unitary and
mathematically sound theory of Faraday effect having the same
generality as Roth's theory (i.e. a theory of the conductivity tensor
for electrons subjected to a periodic electric potential and to a
constant magnetic field in the linear response approximation), but
free of its shortcomings. More precisely in the first paper \cite{CNP}
we started by a rigorous derivation of  the transverse component of the
conductivity tensor in the linear response regime for a finite
sample. It is given as the formula \ref{stone5}
below. To proceed further, we employ a method going back at least 
Sondheimer  and Wilson 
\cite{SW} and which has been also used in the rigorous study of the
Landau magnetism \cite{ABN, Cor, BC, BCL1, BCL2}. The basic idea is
that the traces involved in computing various physical quantities can
be written as integrals involving Green functions (i.e. integral
kernels in the configuration space of either the resolvent or the
semi-group of the Hamiltonian of the system), which are more robust and
easier to control.

As it stands, the conductivity tensor depends upon the shape of the
sample and of boundary conditions which define the quantum
Hamiltonian. The physical idea of the thermodynamic limit for an
intensive physical quantity is that in the limit of large samples it
approaches a limit which is independent upon the shape of the sample,
boundary conditions etc. The existence of the thermodynamic limit is
one of the basic (and far from trivial) problem of statistical
mechanics (see e.g. \cite{Ru, BraRo}). In \cite{CNP} we took for granted that the 
thermodynamic limit of the the transverse component of the
conductivity tensor exists and the limit is smooth as a function of the magnetic
field strength $B$. Moreover, we assumed that the thermodynamic limit
commuted with taking the derivative with respect to $B$. Under these
assumptions, we gave - among other things - explicit formulas for the
Verdet constant in terms
of zero magnetic field Green functions, free of any divergences.

The proof of the thermodynamic limit, which from the mathematical
point of view is the most delicate part of the theory of the Faraday
effect, was left aside in  \cite{CNP}, and is the content of the current paper.
The mathematical problem behind it is hard due to the singularity
induced by the long range magnetic 
perturbation. Even for a simpler problem involving constant magnetic
field - namely the Landau diamagnetism of free electrons - the existence of the 
thermodynamic limit leading to a correct thermodynamic behavior was a
long standing problem. Naive computations 
led to unphysical and contradictory results (see \cite{ABN} for
historical remarks). Accordingly, the first 
rigorous results came as late as 1975 \cite{ABN} and were based on
various identities expressing the gauge invariance which was crucial
in dealing with the singular terms appearing in the thermodynamic limit. Even 
though the importance of gauge invariance was already highlighted in
\cite{ABN}, an efficient way to implement 
this idea at a technical level was still lacking. 
Only recently a regularized magnetic perturbation theory 
based on factorizing the (singular in the thermodynamic limit)
magnetic phase factor has been fully developed 
in \cite{Cor, CN, CN2, Nen}. This regularized magnetic perturbation theory has been already used in 
\cite{Cor, BC, BCL1, BCL2} in order to prove far reaching
generalizations of the results in \cite{ABN}.  

Coming back to the Faraday effect, we would like to stress that the
object at hand is much more singular than the one encountered in the Landau
diamagnetism. This adds an order of magnitude to the mathematical difficulty 
and requires an elaborate and tedious combination of regularized magnetic perturbation theory with 
techniques like Combes-Thomas exponential decay, trace norm estimates and elliptic regularity. 

We expect that the method developed here in order to control the
thermodynamic limit in the presence of an 
extended magnetic field to be useful in related problems, e.g. 
to obtain an elegant and complete study of the diamagnetism and de
Haas-van Alphen effect for electrons in metals.

The content of the paper is as follows. In the rest of this
Introduction we state the mathematical problem, give the main result
in Theorem \ref{teorema1}, and since the proof is quite long and
technical we will briefly describe the main points.
The other sections are devoted to the proof of our main theorem. 
The core of the proof heavily involving  regularized magnetic
perturbation theory is contained in Section 4. Some technical
estimates about exponential decay with uniform control in the spectral
parameter are given as an Appendix.

\subsection{The main result}

Consider a simply connected open and bounded set $\Lambda_1\subset \R^3$,
which contains the origin. We assume that the boundary $\partial
\Lambda_1$ is smooth. Consider a family of scaled domains 
\begin{equation}\label{prel1}
\Lambda_L=\{\x\in\R^3:\; \x/L\in \Lambda_1\},\; L>1.
\end{equation}
We have the estimates 
\begin{equation}\label{prel2}
{\rm Vol}(\Lambda_L)\sim L^3,\quad {\rm Area}(\partial \Lambda_L)\sim L^2.
\end{equation}

We are interested in the thermodynamic limit, which will mean
$L\rightarrow \infty$, that is $\Lambda_L$ will fill out the whole space. 
The one particle Hilbert space is 
$\mathcal{H}_L:=L^2(\Lambda_L)$. Note that we
include the case $L=\infty$.

The one body Hamiltonian of a non-confined particle, subjected to a
constant magnetic field $(0,0,B)$, in an external
potential $V$, formally looks like this:
\begin{align}\label{feshch}
H_\infty(B)&={\bf P}^2(B) +V,
\end{align}
 with
\begin{equation}\label{impuls}
{\bf P}(B)=-i\nabla +B {\bf a}={\bf P}(0)+B{\bf a}.
\end{equation}
Let us explain the various terms. Here ${\bf a}(\x)$ is a smooth
magnetic vector potential which generates a magnetic field
of intensity $B=1$ i.e. $\nabla \wedge {\bf a}(\x)=(0,0,1)$. A 
frequently used magnetic vector potential is the symmetric gauge:
\begin{equation}\label{simg}
{\bf A}(\x)=\frac{1}{2}\n_{3}\wedge \x=(-x_2/2,x_1/2,0),
\end{equation}
where $\n_{3}$ is the unit vector along $z$ axis. We neglect the spin 
structure since it only complicates the notation and does not influence the 
mathematical problem.  

On components, \eqref{impuls} reads as:
\begin{equation}\label{impuls2}
P_j(B)=D_j +Ba_j =:P_j(0)+B a_j,\quad j\in\{1,2,3\}.
\end{equation}

We will from now on assume that $V$ is a $C^\infty(\R^3)$ function,
periodic with respect to the lattice $\mathbb{Z}^3$. Standard
arguments then show that $H_\infty(B)$ is essentially self-adjoint on
$C_0^\infty(\R^3)$. 

When $L<\infty$ we need to specify a boundary condition. We will only
consider Dirichlet boundary conditions, that is we start with the same 
expression as in \eqref{feshch}, defined on
$C_0^\infty(\Lambda_L)$, and we define
$H_L(B)$ to be the Friedrichs extension of it. This is indeed
possible, because our operator can be written as 
$ -\Delta_D +W,$ where $\Delta_D$ is the Dirichlet Laplacian and $W$
is a first order differential operator, relatively bounded
to $-\Delta_D$ (remember that $L<\infty$).  The form domain of
$H_L(B)$ is the Sobolev space $H_0^1(\Lambda_L)$, while the operator
domain (use the estimates in section 10.5, Lemma 10.5.1, in
\cite{Hoe}) is $H^2(\Lambda_L)\cap H_0^1(\Lambda_L)$. Moreover,
$H_L(B)$ is essentially self-adjoint on 
$C_{(0)}^\infty (\overline{\Lambda_L})$, i.e. functions with support in
$\overline{\Lambda_L}$ and indefinitely differentiable in $\Lambda_L$
up to the boundary. 

Another important operator is $(-i\nabla +B {\bf a})_D^2$, i.e. the
usual free magnetic Schr\"odinger operator defined with Dirichlet
boundary conditions. We know that its spectrum is non-negative for all
$L>1$. By adding a positive constant, we can always assume that the
spectrum of $H_L(B)$ is non-negative, uniformly in $L>1$.

Let us now introduce the physical quantity we want to study. Consider 
$\omega\in\mathbb{C}$ and $\Im(\omega)<0$. For some fixed
$\mu\in\R$ and $\beta>0$, define the 
Fermi-Dirac function on its maximal domain of analyticity:
\begin{equation}\label{fd}
f_{FD}(z)=\frac{1}{e^{\beta(z-\mu)}+1}.
\end{equation}
Define
\begin{equation}\label{bo}
d:=\min\left\{ \frac{\pi}{2\beta},\frac{|{\rm Im}\;\omega|}{2}\right\},
\end{equation}
and introduce a counter-clockwise oriented contour given by
\begin{equation}\label{contur1}
\Gamma_\omega=\left \{x\pm i d:\;a\leq x
  <\infty
\right \}\bigcup 
\left \{a+i y:\:-d\leq y \leq
d\right \} 
\end{equation}
where $a+1$ lies below the spectrum of $H_L(B)$. By adding a positive
constant to $V$, we can take $a=-1$ uniformly in $L\geq 1 $ and $B\in [0,1]$. 

We introduce the transverse component of the conductivity tensor (see 
\cite{Roth, CNP}) as 
\begin{align}\label{stone5}
&  \sigma_{L}(B)= -\frac{1}{ {\rm Vol}(\Lambda_L)} \\
&\cdot {\rm Tr}
\int_{\Gamma_\omega}
{f}_{FD}(z) 
\left \{ P_1(B)(H_L(B)-z)^{-1}P_2(B) 
(H_L(B)-z-\omega)^{-1}\right .\nonumber \\
&+ \left . P_1(B)(H_L(B)-z+\omega)^{-1}P_2(B) 
(H_L(B)-z)^{-1} \right \}dz.
\nonumber 
\end{align}

Here {\it Tr} assumes that the integral is a trace-class operator. Now
we are prepared to formulate our main result.

\begin{theorem}\label{teorema1}
The above defined transverse component of the conductivity 
tensor admits the thermodynamic
limit; more precisely:

{\rm i.} The following operator, defined by a $B(L^2(\Lambda_L))$-
norm convergent Riemann integral,
\begin{align}\label{efel}
F_L&:=\int_{\Gamma_\omega}{f}_{FD}(z) \{ P_1(B)(H_L(B)-z)^{-1}P_2(B) 
(H_L(B)-z-\omega)^{-1}\nonumber \\
&+ P_1(B)(H_L(B)-z+\omega )^{-1}P_2(B) 
(H_L(B)-z)^{-1}\}dz,
\end{align} 
is in fact trace-class;

{\rm ii.} Consider the operator $F_\infty$ defined by the same integral but
with $H_\infty(B)$ instead of $H_L(B)$, and defined on the whole
space. Then $F_\infty$ is an integral
operator, with a kernel $\mathcal{F}(\x,\x')$ which is jointly
continuous on its variables. Moreover, the continuous function defined by 
$\R^3\ni \x\to s_{B}(\x):=\mathcal{F}(\x,\x)\in \R$
is periodic with respect to $\mathbb{Z}^3$; 

{\rm iii.} Denote by $\Omega$ the unit cube in $\R^3$. 
The thermodynamic limit exists: 
\begin{align}\label{efel22}
\sigma_\infty(B):=\lim_{L\to \infty}\sigma_{L}(B)=-\int_\Omega s_{B}(\x)
d\x.
\end{align}
Moreover, the mapping $B\to s_{B}\in L^\infty(\Omega)$ is
differentiable at $B=0$ and:
\begin{align}\label{efel2}
\left . \frac{}{}\partial_B\sigma_\infty(0)=-\int_\Omega
\partial_Bs_{B}\right \vert _{B=0}(\x)d\x=\lim_{L\to \infty} 
\partial_B\sigma_{L}(0).
\end{align} 
\end{theorem}

\noindent {\bf Remark} 1. The formula \eqref{efel22} is only the
starting point 
in the study of the Faraday rotation. A related problem is the diamagnetism 
of Bloch electrons, where the main object is the integrated density of states 
of magnetic Schr\"odinger operators (see \cite{HeSj, HeSj2, Ift}). For a systematic treatment of  
magnetic pseudo-differential operators which generalizes our  magnetic
perturbation theory, see \cite{IMP, MP1, MP2, MP3}.

\noindent {\bf Remark} 2. The Dirichlet boundary conditions are
important for us. Even though we suspect that our main result should also
hold for Neumann conditions and for less regular domains (see
\cite{BF, Fou}), we do not see an easy way to prove it. 

\noindent{\bf Remark} 3. We believe that the method we use in
the proof of \eqref{efel2} can be used in order to obtain a stronger result:
the mapping $B\to s_{B}\in L^\infty(\Omega)$ is smooth and for any $n\geq 1$:
\begin{align}\label{efel2bis}
 \frac{}{}\partial_B^n\sigma_\infty(B)=-\int_\Omega
\partial_B^ns_{B}(\x)d\x=\lim_{L\to \infty} 
\partial_B^n\sigma_{L}(B).
\end{align} 
We leave this statement as an open problem. In the rest of the paper
we give the proof of Theorem \ref{teorema1}.

\subsection{A short description of the proof strategy}

Since the proof is rather long, we list here the
main steps and ideas. Let us start with some general
considerations about the thermodynamic limit.

If we are interested in the thermodynamic limit of a quantum physical quantity, the object we need to control is the trace of the operator representing the corresponding quantity. The basic ideea consists in writing this trace as an integral  of the diagonal value of the operator's integral kernel (Schwartz kernel) over the confining box. This procedure makes the quantum thermodynamic limit look very similar to what happens in classical statistical mechanics. More precisely, we need to show that the difference between the integral kernel for the finite box and the one for the entire space decays sufficiently  rapidly with the distance from the boundary of the box, so that the replacement of the integral kernel for the box with the one corresponding to the entire space gives an error term increasing slower than the volume, which then disappears in the limit. 

It turns out that for the transverse conductivity this kernel is far more
complicated than say the heat kernel - whose behavior has been extensively
studied in the literature. Notice for example, that the integrand in  \eqref{stone5} contains  two resolvents sandwiched with magnetic momentum operators. Thus we need a good control of their integral kernels, in particular when the distance between their arguments increases to infinity, and all that uniformly in the spectral parameter $z$. Since a constant magnetic field is present, the biggest difficulty is to deal with the linear growth of the vector potential. Here, the use of  gauge covariance is crucial.

Now let us list the main ideas of the proof. 
For the first statement of the theorem one simply uses integration by parts with respect to
  $z$ in order to transform the integrand into a product of
  Hilbert-Schmidt operators. 

The proof of the second statement is based on elliptic regularity. 
The main technical difficulty is to control the $z$ behavior
  of all our bounds, especially the exponential localization of the magnetic
  resolvents sandwiched with momentum operators. We also have to control the linear growth of the magnetic
  potential. We turn the operator norm
bounds which we obtain in the Appendix into pointwise bounds for certain
integral kernels. It is a long road using magnetic perturbation
theory, but nevertheless, we use nothing more than well-known Combes-Thomas exponential
bounds, local gauge covariance, the Cauchy-Schwarz inequality, and integration by parts.

The third statement of Theorem \ref{teorema1} contains the main
  result and is proved in Section 4. We start with a bit strange
  three-layered partition of unity defined in
  \eqref{margine1}-\eqref{margine111}. This idea goes back
  at least to \cite{BCD, BC}. The main effort consists in 
  isolating the bulk of $\Lambda_L$ - where only operators defined in the
  whole space will act - from the region close to the
  boundary. Note that in the absence of the magnetic field, it would be enough to work with only two cut-off functions: 
one isolating the bulk from the boundary, and the other one supported in  a tubular neighborhood of the boundary. 
When long-range magnetic fields are present, this is not enough. The tubular neigborhood needs to be 
chopped up in many small pieces, in order to apply local gauge transformations (see below why we need them).  

The central idea in proving \eqref{efel22} is to show 
that the contribution to the total trace of the region
close to the boundary grows slower than
the volume. Technically, this is obtained by approximating the true
resolvent $(H_L(B)-z)^{-1}$ with an operator $U_L(B,z)$ given in
\eqref{geigi4}. $U_L(B,z)$ contains the bulk term, plus a boundary
contribution which consists from a sum of terms each {\it locally}
approximating $(H_L(B)-z)^{-1}$ and containing a specially tailored local gauge
given in \eqref{geigi}. These locally defined vector potentials are
made {\it globally} bounded with the help of our third layer of cut-off
functions $\tilde{\tilde{g}}_\gamma$'s. The switch to the local gauge
is performed through the central identity \eqref{fazam2}. An explanation of why
$U_L(B,z)$ is a convenient approximation for the full resolvent can be
found right after \eqref{geigi6}, and there is the place where we first
fully use the exponential estimates of the Appendix. 

In Proposition \ref{prop22} we prove that we can replace
$(H_L(B)-z)^{-1}$ with $U_L(B,z)$ in the conductivity formula, without
changing the value of the limit. In Proposition \ref{prop22} we show
that only the bulk term from $U_L(B,z)$ will contribute. In
Proposition \ref{prop222} we show that removing the cut-offs only
gives a surface contribution. 

The proof of \eqref{efel2} is heavily based on magnetic
  perturbation theory. Although the technical estimates are
  considerably more involved than at the previous point, the main idea
  is the same: the boundary terms can be discarded. The full power of
  the magnetic phases is used in Lemma \ref{lema30}; an heuristic
  explanation of how and why they manage to kill the linear growth of
  the magnetic potential is given right after \eqref{fazam99}.

\section{Proof of {\rm i}.} 

The integrand in \eqref{efel} is a bounded operator, with an $L^2$
norm bounded by a constant times $|\Re(z)|^2$, uniformly in $L$, as we
can see from \eqref{impuls2} and \eqref{keyest2}. Because $f_{FD}$ has
an exponential decay in $|\Re(z)|$, the integral converges and defines
a bounded operator. Let us note that the integrand is not a trace
class operator under our conditions. But the total integral is a
different matter. The point is that we can integrate by parts with
respect to $z$ by using anti-derivatives of $f_{FD}$ which are still
decaying exponentially at infinity. By doing this at least four
times, we obtain integrals
of the form

\begin{equation}\label{integrparti}
\int_{\Gamma_\omega}
\tilde{f}(z) P_1(B)(H_L(B)-z)^{-m}P_2(B) 
(H_L(B)-z-\omega)^{-n}dz
\end{equation} 
where $m+n\geq 5$, hence $\max\{m,n\}\geq 3$. Assume that $m\geq
3$. Then we can write the above integral as 
\begin{equation}\label{integrparti2}
\int_{\Gamma_\omega}
\tilde{f}(z) P_1(B)(H_L(B)-z)^{-m+2}(H_L(B)-z)^{-2}P_2(B) 
(H_L(B)-z-\omega)^{-n}dz.
\end{equation}
The main point is that $(H_L(B)-z)^{-1}$ is Hilbert-Schmidt since we
can write 
\begin{align}\label{sssa}
&(H_L(B)-z)^{-1} \nonumber \\
&=[(-i\nabla +B{\bf a})_D^2+1]^{-1} \left \{[(-i\nabla
  +B{\bf a})_D^2+1](H_L(B)-z)^{-1}\right \},
\end{align}
and by using \eqref{expdek3} with $\delta=0$, and \eqref{intkern}, we obtain
\begin{align}\label{ssshs}
||(H_L(B)-z)^{-1}||_{B_2}\leq \C \cdot \sqrt{{\rm Vol}(\Lambda_L)} 
\langle \Re(z)\rangle ,
\end{align}
where the above constant does not depend on $L$ and $z$. Thus
$(H_L(B)-z)^{-2}$ is trace class, and the trace norm of the integrand
in \eqref{integrparti2} is bounded by 
$$\C \cdot |\tilde{f}(z)|\cdot \langle \Re(z)\rangle^4 
\cdot  {\rm Vol}(\Lambda_L)$$
where again the constant is uniform in $L$ and $z$. This now is 
integrable on the contour, thus the integral defines a trace class
operator. Moreover, its trace grows at most like the volume of
$\Lambda_L$, hence 
$\limsup_{L\to\infty}|\sigma_L|<\infty$.\qed 

\vspace{0.5cm}

{\bf Remark.} The same type of argument may be used to show that
$\sigma_L(B)$ is smooth in $B$, by repeatedly using the formal identity 
\begin{align}\label{derivrez}
\partial_B(H_L(B)-z)^{-1}=-(H_L(B)-z)^{-1}\{\partial_BH_L(B)\}(H_L(B)-z)^{-1}
\end{align}
in the sense of bounded operators. Note the important fact that
$\partial_BH_L(B)$ will generate some linear growing terms coming
from the magnetic vector potential ${\bf a}(\x)$, therefore the trace
norm of the new integrand will grow like $L^4$. We therefore cannot
conclude here that the derivatives $|\partial_B^n\sigma_L(B)|$ will
admit a finite $\limsup$ when $L\to\infty$.

\section{Proof of {\rm ii}.}

We are going to prove the regularity statement for the kernel without
using the periodicity of $V$, only the fact that the potential is
smooth and bounded on $\R^3$ together with all its derivatives.

The strategy consists in integrating by parts with respect to $z$ many
times, such that we obtain high powers of the resolvent
$(H_\infty(B)-\zeta)^{-N}$. Then we will prove that
$P_j(B)(H_\infty(B)-\zeta)^{-N}P_k(B)$ has a smooth kernel which 
does not grow too fast with $\langle |\zeta|\rangle $. 

Let us now be more precise and start with some technical results. 

\begin{proposition}\label{prop5}
Fix $0<\eta<1$ and choose $z\in\mathbb{C}$ with ${\rm dist}\{z,[0,\infty)\}= \eta>0 $. Let 
$r=\langle |\Re(z)|\rangle $. Then the operator
$ (H_\infty(0)-z)^{-1}$
 has an integral kernel $G_1(\x,\x';z)$ which is smooth away
 from the diagonal $\x=\x'$. There exists $\delta >0$ and some 
$M\geq 1$ such that for
 any multi-index $\alpha\in\mathbb{N}^3$ with $|\alpha|\leq 1$ we have 
the estimate 

\begin{align}\label{intkernrez}
\sup_{\x\neq \x'\in\R^3}|\x-\x'|^{|\alpha|+1}e^{\frac{\delta}
{\langle r\rangle}|\x-\x'|}
|D_{\x}^{\alpha}G_1(\x,\x';z)| =C_1(\alpha,\eta)\langle r\rangle ^M<\infty.
\end{align}
\end{proposition}

\noindent{Proof.} The result without the exponential decay is
essentially contained in \cite{Sj}. 
The symbol of $H_\infty(0)$, denoted by $h_0(\x,\xi)$ belongs to
$S^2_{1,0}(\R^3\times \R^3)$, $H_\infty(0)\in L^2_{1,0}(\R^3)$
(see Example 3.1 in \cite{Sj}), and is uniformly elliptic.

Fix $\lambda >0$ large enough. We can apply Theorem 4.1 in \cite{Sj} 
and construct a parametrix for
$H_\infty(0)+\lambda $ starting from the symbol
$q_0(\x,\xi):=1/(h(\x,\xi)+\lambda )\in
S^{-2}_{1,0}(\R^3\times \R^3)$. The symbol giving the parametrix is an
asymptotic sum of symbols, starting with $q_0$, then the next one is
in $S^{-3}_{1,0}$ and so on. Each term gives a contribution to the
integral kernel of the parametrix. The most singular contribution is
(in the sense of oscillatory integrals):
$$\frac{1}{(2\pi)^3}\int_{\R^3}e^{i(\x-\x')\cdot\xi}q_0(\x,\xi)d\xi.$$ 
We only have to consider a few terms besides this one, 
since symbols in $S^{-N}_{1,0}$
for large $N$ generate more and more regular kernels at the diagonal. 
By standard ``integration by parts'' tricks, and using the fact that
we work in three dimensions, one can prove the
estimate 
\begin{align}\label{intkernrez2}
\sup_{\x\neq \x'\in\R^3}|\x-\x'|^{|\alpha|+1}
|D_{\x}^{\alpha}G_1(\x,\x';-\lambda )| =\C (\alpha,\lambda 
)<\infty.
\end{align}
In fact, outside the region $|\x-\x'|\geq 1$ we can integrate by parts
several times with respect to $\xi$ and prove that
$G_1(\x,\x';-\lambda )$ decays faster than  any power of 
$\langle|\x-\x'|\rangle $. But the Combes-Thomas method will give a
better, exponential localization. 

The important thing is that the $L^2$ estimates from
the Combes-Thomas argument can be transferred into point-wise
estimates for the kernel. Let us now prove this. 

Using \eqref{intkern}at $L=\infty$ and $B=0$, 
together with the triangle and Cauchy-Schwarz 
inequalities, we get that $(-\Delta+\lambda)^{-1}$ with exponential 
weights maps $L^2$ into
$L^\infty$. The key estimate is ($0<c<1$)
$$e^{-c\sqrt{\lambda}|\x
  -\x_0|}|K_\infty(\x,\x')|e^{c\sqrt{\lambda}|\x' -\x_0|}\leq 
\frac{e^{-(1-c)\sqrt{\lambda}|\x
  -\x'|}}{4\pi|\x-\x'|}.$$

Since we can write:
\begin{align}\label{nou1}
(H_\infty(0)+\lambda)^{-1}=(-\Delta+\lambda)^{-1}
(-\Delta+\lambda) (H_\infty(0)+\lambda)^{-1},
\end{align}
using \eqref{expcor2} (at $L=\infty$ and $B=0$), it
follows that this resolvent with exponential weights is a
bounded map from $L^2$ into $L^\infty$. More
precisely, for any $\x_0\in\R^3$, there exists $0<c<1$ small enough such
that:
\begin{align}\label{intkernrez5}
\sup_{\x_0}||e^{-c \sqrt{\lambda}\langle \cdot-\x_0 \rangle} (H_\infty(0)+\lambda)^{-1}e^{c\sqrt{\lambda}\langle
  \cdot -\x_0\rangle}||_{B(L^2,L^\infty)}\leq \C,\quad \lambda\geq \lambda_0.
\end{align} 
Now if we look at the map 
$$C_0^\infty(\R^3) \ni\Psi\to
\int_{\R^3}G_1(\x_0,\x;-\lambda )
e^{\delta_\lambda\langle \x-\x_0\rangle }\Psi(\x)d\x,$$
(it makes sense to fix $\x_0$ since the resolvent maps smooth
functions into smooth functions), we see that by using
\eqref{intkernrez5} we can extend this map to a linear and bounded 
functional on $L^2$. Riesz' representation theorem then gives: 
\begin{align}\label{intkernrez77}
\sup_{\x_0\in\R^3}||e^{c\sqrt{\lambda}\langle \cdot-\x_0\rangle
}G_1(\x_0,\cdot;-\lambda
)||_{L^2}=\sup_{\x_0\in\R^3}||e^{c\sqrt{\lambda}\langle
  \cdot-\x_0\rangle }G_1(\cdot,\x_0;-\lambda )||_{L^2}\leq \C ,
\end{align}
uniformly in $\lambda\geq \lambda_0$. 
Using this, together with the Cauchy-Schwarz and the triangle inequality, 
we get that the
integral kernel $G_2(\x,\x';-\lambda )$ of
$(H_\infty(0)+\lambda)^{-2}$ obeys uniformly in
$\lambda\geq \lambda_0$:
\begin{align}\label{intkernrez773}
&\sup_{\x,\x'}e^{c\sqrt{\lambda} |
  \x-\x'| }|G_2(\x,\x';-\lambda )|\\
&\leq \sup_{\x,\x'}\int_{\R^3}|e^{c\sqrt{\lambda}\langle
  \x-\x''\rangle }G_1(\x,\x'';-\lambda )e^{c\sqrt{\lambda}\langle
  \x''-\x'\rangle }G_1(\x'',\x';-\lambda )|d\x''\leq \C
.\nonumber 
\end{align}

Now if $|\x-\x'|\geq 1$, write 
\begin{align}\label{intkernrez3w}
G_1(\x,\x';-\lambda )=\int_{\lambda}^{\infty}G_2(\x,\x';-\lambda_1 )d\lambda_1,
\end{align}
which together with \eqref{intkernrez773} and the integrability of
$e^{-\frac{c}{2}\sqrt{\lambda_1}}$ imply that 
\begin{align}\label{intkernrez774}
\sup_{|\x-\x'|\geq 1 }e^{\frac{c}{2}\sqrt{\lambda} |
  \x-\x'| }|G_1(\x,\x';-\lambda )|\leq \C.
\end{align}

We can also deal with derivatives with respect to $\x$, by showing that the
operator $D_j (H_\infty(0)+\lambda)^{-N}$ ($N$ large enough) has an
integral kernel $D_jG_N(\x,\x';-\lambda )$ obeying the same type of
estimate as in \eqref{intkernrez773}. This is done by commuting $D_j$
several times with a few resolvents; let us see how it works.  

First, by commuting we have:
\begin{align}\label{revised2}
& D_j
(H_\infty(0)+\lambda)^{-1}=(H_\infty(0)+\lambda)^{-1}D_j+(H_\infty(0)+\lambda)^{-2}T_1,\\
&T_1=[H_\infty(0),D_j]-[H_\infty(0),[H_\infty(0),D_j]](H_\infty(0)+\lambda)^{-1},\nonumber 
\end{align}
where $T_1$ is bounded due to the fact that $[H_\infty(0),D_j]$ is bounded, while $[H_\infty(0),[H_\infty(0),D_j]]$ is relatively bounded with respect to $H_\infty(0)$. 

Second, by commuting twice we have:
\begin{align}\label{revised3}
D_j
(H_\infty(0)+\lambda)^{-2}=(H_\infty(0)+\lambda)^{-2}D_j+(H_\infty(0)+\lambda)^{-3}T_1+(H_\infty(0)+\lambda)^{-2}T_1(H_\infty(0)+\lambda)^{-1}.
\end{align}

This identity allows us to write:
\begin{align}\label{revised1}
& D_j
(H_\infty(0)+\lambda)^{-N+2}=(H_\infty(0)+\lambda)^{-2}T(H_\infty(0)+\lambda)^{-N+5},\\
&T=\{D_j+(H_\infty(0)+\lambda)^{-1}T_1+T_1(H_\infty(0)+\lambda)^{-1}\}(H_\infty(0)+\lambda)^{-1},\nonumber 
\end{align}
where $T$ with exponential weights is bounded from $L^2$ to
$L^2$.  
Then we prove that the integral kernel of $D_j
(H_\infty(0)+\lambda)^{-N+2}$ obeys an $L^2$ estimate like
in \eqref{intkernrez77}, then from the
identity
$$ D_j
(H_\infty(0)+\lambda)^{-N}= D_j
(H_\infty(0)+\lambda)^{-N+2}(H_\infty(0)+\lambda)^{-2}$$
we get the needed $L^\infty$ estimate by mimicking
\eqref{intkernrez773}.

Then we write 
$$D_j G_1(\x,\x';-\lambda
)=\frac{(-1)^{N}}{(N-1)!}\int_{\lambda}^{\infty}d\lambda_1\int_{\lambda_1}^{\infty}
d\lambda_2\dots  \int_{\lambda_{N-1}}^{\infty}d\lambda_ND_j G_N(\x,\x';-\lambda' )$$
and propagate the exponential decay over the integrals in
$\lambda$.

Therefore we can state the first result regarding the exponential
localization. For $\lambda$ large enough we have:
\begin{align}\label{intkernrez3}
\sup_{\x\neq \x'\in\R^3}|\x-\x'|^{|\alpha|+1}e^{|\x-\x'|}
|D_{\x}^{\alpha}G_1(\x,\x';-\lambda )| =\C(\alpha,\lambda 
)<\infty.
\end{align}

Now let us investigate the $z$ dependence. Let us apply the resolvent
identity several times and get ($N\geq 2$):
\begin{align}\label{intkernrez4}
(H_\infty(0)-z)^{-1}&=(H_\infty(0)+\lambda)^{-1}+(z+\lambda)
(H_\infty(0)+\lambda)^{-2}+\dots\nonumber
\\
&+(z+\lambda)^N(H_\infty(0)+\lambda)^{-N}(H_\infty(0)-z)^{-1}.
\end{align}
The idea is to keep the $z$ dependence to the right in the last term, and to keep a regular kernel to the 
left. We start with a norm estimate. From the usual resolvent
identity:
$$(H_\infty(0)-z)^{-1}=(H_\infty(0)+\lambda)^{-1}+(z+\lambda)
(H_\infty(0)+\lambda)^{-1}(H_\infty(0)-z)^{-1},$$
the use \eqref{intkernrez77} and \eqref{expdek} (with
$L=\infty$) provides us  
with some $N_1>N$ such that:
\begin{align}\label{intkernrez6}
\sup_{\x_0}||e^{-\frac{\delta}{r}\langle \cdot-\x_0 \rangle}
 (H_\infty(0)-z)^{-1}
e^{\frac{\delta}{r}\langle \cdot-\x_0 \rangle} ||_{B(L^2,L^\infty)}
\leq \C(\eta)\cdot r^{N_1}.
\end{align}

This estimate implies that the map 
(initially defined on compactly supported functions)
$$L^2(\R^3) \ni\Psi\to
\int_{\R^3}G_1(\x_0,\x;\overline{z} )
e^{\frac{\delta}{r}\langle \x-\x_0\rangle }\Psi(\x)d\x\; \in \mathbb{C},$$
is a bounded linear functional. Riesz' representation theorem leads us to: 
\begin{align}\label{intkernrez7}
\sup_{\x_0\in\R^3}||e^{\frac{\delta}{r}\langle \cdot-\x_0\rangle }G_1(\x_0,\cdot;\overline{z} )||_{L^2}=\sup_{\x_0\in\R^3}||e^{\frac{\delta}{r}\langle \cdot-\x_0\rangle }G_1(\cdot,\x_0;z )||_{L^2}\leq \C(\eta)\cdot r^{N_1}.
\end{align}
We are only left with the case in which we have a derivative on the
left. Using \eqref{intkernrez7} in 
\eqref{intkernrez4} and the other results we have obtained for the
kernel where $z=\lambda$, it is not hard to 
obtain the exponential decay claimed in \eqref{intkernrez}. \qed 

\vspace{0.5cm}

\begin{proposition}\label{prop6}
Assume that $\alpha\in\{0,1\}$ and $0<\eta<1$. Let ${\rm
  dist}\{z,[0,\infty)\}=\eta$ and $N$ is large enough. Then 
the operator $ P_1^\alpha(B)(H_\infty(B)-z)^{-N}P_2^{1-\alpha}(B)$ has a 
jointly continuous
integral kernel $K_{N,B}(\x,\x';z)$, and there exists $\delta>0$ small
enough and $M$ large enough such that:
\begin{equation}\label{31415}
\sup_{\x,\x'\in \R^3}e^{\frac{\delta}{\langle r\rangle}|\x-\x'|}
|K_{N,B}(\x,\x';z)|\leq
const(N,B,\eta)\cdot \langle r\rangle ^M.
\end{equation}
\end{proposition}

\noindent{\bf Proof.} Although this particular result might be
obtained with other methods, we will employ the magnetic perturbation
theory as developed in \cite{CN, BC, Nen}. For different approaches involving 
magnetic pseudo-differential calculus, 
see \cite{HeSj, HeSj2, MP1, MP2, IMP}.

We can assume that the magnetic vector potential is expressed in the
transverse gauge, given in \eqref{simg}. Define the antisymmetric magnetic 
gauge
phase 
\begin{align}\label{fazam}
\varphi_0(\x,\y):=-{\bf A}(\y)\cdot \x=\frac{1}{2} (y_2 x_1 - y_1
x_2)=\frac{1}{2}{\bf e}_3\cdot (\x\wedge \y),
\end{align}
where ${\bf e}_3$ denotes the unit vector $(0,0,1)\in \R^3$. 
Then we have the following identity 
(true for instance on Schwartz functions), 
valid for every vector $\y$ kept fixed in $\R^3$:
\begin{equation}\label{fazam2}
\{{\bf P}_\x(0)+B {\bf
  A}(\x)\}e^{iB\varphi_0(\x,\y)}=e^{iB\varphi_0(\x,\y)}\{{\bf P}_\x(0)+B
{\bf A}(\x -\y )\}.
\end{equation}
For every $z\in \mathbb{C}\setminus [0,\infty)$ define
\begin{equation}\label{fazam3}
\mathcal{S}_B(\x,\x';z):=e^{iB\varphi_0(\x,\x')}G_1(\x,\x';z ). 
\end{equation}
This integral kernel generates an $L^2$ bounded operator (use
the estimate \eqref{intkernrez} and employ the Schur-Holmgren criterion). We denote this operator by
$S_B(z)$. There are two important facts related to this kernel. 
First, 
$G_1(\x,\x';z )$ solves the distributional equation
$$({\bf P}_\x^2(0)+V(\x)-z)G_1(\x,\x';z )=\delta(\x -\x').$$ 
Second, one can prove that $S_B(z)$ maps Schwartz 
functions into Schwartz functions. 
Moreover, employing \eqref{fazam2} and integrating by parts, we can establish an identity which holds 
at first in 
the weak sense on Schwartz functions:
\begin{equation}\label{fazam4''}
\langle(H_\infty(B)-\overline{z})\Psi,S_B(z)\Xi\rangle =\langle \Psi, (1+B\;T_B(z))\Xi\rangle ,
\end{equation}
where $T_B(z)$ is the operator generated by the following integral kernel:
\begin{align}\label{fazam4}
&\mathcal{T}_B(\x,\x';z):=e^{iB\varphi_0(\x,\x')} \\ 
&\cdot \left\{ \frac{}{} -2i {\bf A}(\x
  -\x' )\cdot \nabla_\x G_1(\x,\x';z )+ B|{\bf A}(\x-\x'
  )|^2G_1(\x,\x';z ) \right \}.\nonumber 
\end{align}
Using \eqref{intkernrez}, and the fact that $|{\bf A}(\x-\x'
  )|\leq |\x-\x'|$, we obtain the following point-wise estimate true
  for all $\x\neq \x'$:
\begin{align}\label{fazam5}
\max\{ |\mathcal{S}_B(\x,\x';z)|,|\mathcal{T}_B(\x,\x';z)|\}\leq \C 
\frac{e^{-\frac{\delta}{\langle r\rangle}|\x-\x'|}}{|\x-\x'|}\cdot 
\langle r\rangle ^M .
\end{align}
Clearly,  $T_B(z)$ can be extended to a bounded operator. 

Now let us prove that \eqref{fazam4''} also holds true in the strong sense. Because $H_\infty(B)$ is 
essentially self-adjoint on 
the set of Schwartz functions,  \eqref{fazam4''} can be extended to any $\Xi\in L^2$ and any 
$\Psi\in {\rm Dom}(H_\infty(B))$. It means that the range of $S_B(z)$ belongs to the domain of 
$H_\infty(B)$ and 
\begin{equation}\label{fazam4'}
(H_\infty(B)-z)S_B(z) =1+B\;T_B(z).
\end{equation}

At this point we can establish the following identity, valid for all $z$ of
interest:
\begin{align}\label{fazam6}
(H_\infty(B)-z)^{-1}&=S_B(z)- B(H_\infty(B)-z)^{-1}T_B(z).
\end{align}
Denote by $\hat{T}_B(z):=T^*_B(\overline{z})$ the bounded operator generated by the kernel 
$$\hat{\mathcal{T}}_B(\x,\x';z):=\overline{\mathcal{T}_B(\x',\x;\overline{z})}.$$ Now replace $z$ with $\overline{z}$ in \eqref{fazam6} and then take the adjoint. We obtain:
\begin{align}\label{fazam667}
(H_\infty(B)-z)^{-1}&=S_B(z)- B\hat{T}_B(z)(H_\infty(B)-z)^{-1}.
\end{align}

Denote by $K_{1,B}(\x,\x';z)$ the integral kernel of
$(H_\infty(B)-z)^{-1}$. With \eqref{fazam667} as starting point, together with \eqref{expdek}, we can use the same argument as in the zero magnetic field case, in order to show that the resolvent 
(with exponential weights) maps 
$L^2$ into $L^\infty$. Thus its kernel obeys an estimate like in \eqref{intkernrez7}. Moreover, 
\eqref{fazam667} implies via the Cauchy-Schwarz inequality
that  
\begin{align}\label{fazam7}
|K_{1,B}(\x,\x';z)|\leq \C
\frac{e^{-\frac{\delta}{\langle r\rangle}|\x-\x'|}}{|\x-\x'|}\cdot 
\langle r\rangle ^M
\end{align}
Moreover, we can repeat the arguments from the proof of Proposition
\ref{prop5} and get the boundedness and joint continuity for the
kernel of $P_j(B)(H_\infty(B)-z)^{-N}$ for large $N$, since all we
have to do is to change $D_j$ with  $P_j(B)$ and to notice that the
formal commutator $[P_j(B),H_\infty(B)]$ is only linear in
$P_k(B)$'s, therefore  $P_j(B)$ ''commutes well'' with
$(H_\infty(B)-z)^{-1}$. 

Now let us show that $(H_\infty(B)-z)^{-1}$ maps $L^2$ into H\"older 
continuous functions. In fact, one can prove the
following estimate: 

\begin{lemma}\label{lemahoelder} Fix a compact set
$U\subset\R^3$, and fix $\beta\in (0,1/2)$. Take 
$\psi$ with $||\psi||_{L^2}=1$. 
Then there exist two positive constants $C$ and $M$ such that 
\begin{align}\label{fazam8}
&\sup_{z\in \Gamma_\omega}\langle r\rangle ^{-M}
|\{(H_\infty(B)-z)^{-1}\psi\}(\x)-\{(H_\infty(B)-z)^{-1}\psi\}(\y)|
\nonumber
\\
&\leq C \cdot |\x-\y|^\beta,
\end{align}
for any $\x,\y\in U$. The same estimate holds true for $S_B(z)$.
\end{lemma}
\noindent{\bf Proof.}  The domain of $H_\infty(B)$ is locally $H^2$, 
hence for some positive $\lambda$, the function 
$(H_\infty(B)+\lambda)^{-1}\psi$ is locally $H^2$. From the Sobolev 
embedding lemma, we obtain that $(H_\infty(B)+\lambda)^{-1}\psi$ is 
$\beta$-H\"older continuous for every $\beta\in [0,1/2)$. The estimate 
\eqref{fazam8} follows from the resolvent identity  
$$(H_\infty(B)-z)^{-1}=(H_\infty(B)+\lambda)^{-1}+(z+\lambda)
(H_\infty(B)+\lambda)^{-1}(H_\infty(B)-z)^{-1}.$$
The same result for $S_B(z)$ follows from 
\eqref{fazam6}. \qed

\vspace{0.5cm}

We now can start the actual proof of ({\bf ii}). We integrate by parts $N$ 
times with respect to $z$ in the expression of $F_\infty$, 
and $N$ is supposed to be large. The terms we obtain in the integrand will look like this 
one: 
\begin{align}
 P_j(B)(H_\infty(B)-z_1)^{-N_1}P_k(B)(H_\infty(B)-z_2)^{-N_2},
\end{align}
where $N_1+N_2=N+2$, and either $N_1$ or $N_2$ is large. Here $z_1$ and $z_2$ 
are complex numbers like $z\in \Gamma_\omega$ or $z\pm \omega$.

By repeated commutations, we can always write this operator as  
$(H_\infty(B)-z_1)^{-1}W(H_\infty(B)-z_2)^{-1}$, where $W$ is a sum of terms 
like this one:
$$W_1:=(H_\infty(B)-z_1)^{-n_1}v_{\alpha_1}P_{\alpha_1}(B)
(H_\infty(B)-z_1)^{-n_2}v_{\alpha_2}P_{\alpha_2}(B) 
(H_\infty(B)-z_1)^{-n_3}.$$ 
Here $n_1,n_2,n_3\geq 1$, $v_\alpha$'s are smooth and uniformly bounded 
functions. Note that such a term always starts and ends with a resolvent. 
Using the fact that $P_{\alpha_1}(B)
(H_\infty(B)-z_1)^{-n_2}P_{\alpha_2}(B)$ is a bounded operator, then we can 
always write $W$ as the product of the form $(H_\infty(B)-z_1)^{-1}\tilde{W}$ 
where $\tilde{W}$ is a bounded operator whose norm increases at most 
polynomially in $r$. Thus $e^{-\langle \cdot \rangle \delta/r}W$ is a 
Hilbert-Schmidt operator, with a H-S norm which increases polynomially in 
$r$. Hence it is an integral operator which obeys the estimate:
  \begin{align}\label{marsy22}
\int \int 
e^{-2\langle \y \rangle \delta/r}|W(\y,\y';z)|^2d\y\; d\y'\leq \C \; r^M,
\end{align}
where the constant and $M$ are independent of $r$. Thus we have:
\begin{align}\label{marsy23}
&P_j(B)(H_\infty(B)-z_1)^{-N_1}P_k(B)(H_\infty(B)-z_2)^{-N_2}\nonumber \\
&=(H_\infty(B)-z_1)^{-1}
e^{\langle \cdot \rangle \delta/r}e^{-\langle \cdot \rangle \delta/r}
W(H_\infty(B)-z_2)^{-1}.
\end{align}
Using an identity like the one in \eqref{expdek4}, we can rewrite the above 
operator as:
$$e^{\langle \cdot \rangle \delta/r}(H_\infty(B)-z_1)^{-1}
Te^{-\langle \cdot \rangle \delta/r}W(H_\infty(B)-z_2)^{-1},$$
where $T$ is a bounded operator uniformly in $r$ if $\delta$ is small enough. 
Thus $W':=Te^{-\langle \cdot \rangle \delta/r}W$ is Hilbert-Schmidt, and 
has an integral kernel whose norm in $L^2(\R^6)$ is polynomially bounded in 
$r$.  

Therefore we are left with the investigation of the joint continuity in $\x$ 
and $\x'$ of the integral kernel defined by:
$$f_\infty(\x,\x'):=\int \int K_{1,B}(\x,\y;z_1) 
W'(\y,\y')K_{1,B}(\y',\x';z_2)d\y\; d\y',$$
where $\int \int |W'(\y,\y')|^2d\y\; d\y'\leq \C\; r^M$. Using this, 
\eqref{fazam7}, and the Cauchy-Schwarz inequality, we obtain 
$|f(\x,\x')|\leq \C\; r^{M_1} ||W'||_{B_2}$, uniformly in $\x,\x'\in U$. 
Since $W'$ is Hilbert-Schmidt, we can approximate 
$W'(\y,\y')$ with a finite sum of the type $\sum g_j(\y)h_j(\y')$ where $g$'s 
and $h$'s are $L^2$. Thus we can (uniformly in $\x,\x'\in U$) approximate the function $f_\infty$ with 
functions of the type 
$$\sum_{j=1}^n \{(H_\infty(B)-z_1)^{-1}g_j\}(\x)
\overline{\{(H_\infty(B)-z_2)^{-1}\overline{h}_j\}(\x')},$$
which from Lemma \ref{lemahoelder} we know they are continuous on compacts. 
Hence $f_\infty$ is jointly continuous on its variables. As for the integral 
kernel of $F_\infty$, we see that it can be written as the integral with 
respect to $z$ of a finite number of kernels of the same type as $f_\infty$. 
Due to the exponential decay of $f_{FD}$ (see \eqref{fd}), we see that we can 
approximate $F_\infty(\x,\x')$ uniformly on compacts with continuous 
functions, and we are done. 

Therefore the function $s_B$ as defined in Theorem
\ref{teorema1} is a continuous function. If $V$ is periodic with
respect to $\mathbb{Z}^3$, then
$H_\infty(B)$ commutes with the magnetic translations, defined for
every $\gamma\in\mathbb{Z}^3$ as (see also \eqref{fazam2}): 
$$[M_{\gamma}
\psi](\x):=e^{iB\varphi_0(\x,\gamma)}\psi(\x-\gamma),\quad M_{-\gamma}M_{\gamma}=1.$$
Hence we have that as operators, $M_{-\gamma}F_\infty
M_{\gamma}=F_\infty$, which for kernels gives 
\begin{align}\label{fazam12}
e^{-iB\varphi_0(\x,\gamma)}\mathcal{F}(\x+\gamma,\x'+\gamma)e^{iB\varphi_0(\x'+\gamma,\gamma)}\mathcal{F}(\x,\x').
\end{align}
Now since $\varphi_0(\gamma,\gamma)=0$, when we put $\x=\x'$ we get
$s_B(\x+\gamma)=s_B(\x)$ and we are done with (ii). \qed

\section{Proof of {\rm iii}.}

We start by proving the thermodynamic limit for the conductivity, that
is \eqref{efel22}. We need to introduce a special partition of unity
in $\Lambda_L$. 

\subsection{Partitions and cut-offs}

Fix $0<\alpha<1$ (small enough, to be chosen later) and define for
$t>0$:
\begin{equation}\label{margine1}
\Xi_L(t):=\left \{\x\in \overline{\Lambda_L}:\; {\rm
    dist}\{\x,\partial \Lambda_L\}\leq t L^\alpha\right \}.
\end{equation}
This models a ''thin'' compact subset of $\Lambda_L$, near the boundary, with a volume of order $
 t L^{2+\alpha}$. Because we assumed that the boundary
$\partial\Lambda_1$ was smooth, all points of $\Xi_L(t)$ have unique
projections on $\partial\Lambda_L$, if $L$ is large enough.  

Then if $t_1<t_2$ we have $\Xi_L(t_1)\subset \Xi_L(t_2)$ and :
\begin{align}\label{margine3}
{\rm dist}\{\Xi_L(t_1), \overline{\Lambda_L\setminus\Xi_L(t_2)}\}
\geq (t_2-t_1) L^{\alpha}. 
\end{align}

The subset $\overline{\Lambda_L\setminus\Xi_L(1)}$ models the ``bulk
region'' of $\Lambda_L$, which is still ``far-away'' from the
boundary. 

Now consider the inclusion of $\Xi_L(2)$ in the dilated lattice $L^\alpha \mathbb{Z}^3$. That is we cover 
$\Xi_L(2)$ with disjoint closed cubic boxes parallel to the
coordinate axis, centered at points in $L^\alpha \mathbb{Z}^3$, of side length
$L^\alpha$. Denote by $E\subset L^\alpha \mathbb{Z}^3$ the 
set of centers of those cubes which have common points with $\Xi_L(2)$. Clearly, due to volume 
considerations,  $\# E\sim L^{2-2\alpha}$. 

In order to fix notation, let us denote by
$K(\gamma, s)$ the cube centered at $\gamma\in E$, with side length
equal to $s\geq L^\alpha$. Moreover, denote by $\tilde{E}:=E\cup \{(0,0,0)\}$ (note
that the origin cannot belong to $E$ if $L$ is large enough). 

Now choose a partition of unity $\{g_{\gamma}\}_{\gamma\in \tilde{E}}$
of $\Lambda_L$ which has the following
properties: 
\begin{align}\label{margine4}
&{\rm supp}(g_{0})\subset \Lambda_L\setminus\Xi_L(1);\\
&{\rm supp}(g_{\gamma})\subset K(\gamma, 2L^\alpha), \gamma\in
E;\label{margine5}\\
&0\leq g_\gamma\leq 1,\qquad \sum_{\gamma\in
  \tilde{E}}g_\gamma(\x)=1,\quad  \forall \x\in
\Lambda_L;\label{margine6}\\
&||D^\beta g_\gamma||_\infty\sim L^{-\alpha |\beta|},\quad \forall 
\beta\in \mathbb{N}^3,\quad  \forall \gamma\in \tilde{E}\label{margine7}. 
\end{align}
This partition has the property that if we restrict ourselves to $E$,
then uniformly in $L$, the number of $g_\gamma$'s which are not zero
at the same time is bounded by a constant. Only $g_0$ has $\sim L^{2-2\alpha}$ 
neighbors whose supports have common points with ${\rm supp}(g_{0})$. 

Now we choose another set of functions,
$\{\tilde{g}_{\gamma}\}_{\gamma\in \tilde{E}}$ having the following
properties: 
\begin{align}\label{margine8}
&{\rm supp}(\tilde{g}_{0})\subset
\Lambda_L\setminus \Xi_L(1/4);\qquad \tilde{g}_{0}(\x)=1 \: {\rm
  if}\: \x\in \Lambda_L\setminus \Xi_L(1/2);\\
&{\rm supp}(\tilde{g}_{\gamma})\subset K(\gamma, 10 L^\alpha), \gamma\in
E;\quad \tilde{g}_{\gamma}(\x)=1 \: {\rm
  if}\: \x\in K(\gamma, 9 L^\alpha)\label{margine9}\\
&||D^\beta \tilde{g}_\gamma||_\infty\sim L^{-\alpha |\beta|},\quad \forall 
\beta\in \mathbb{N}^3,\quad  \forall \gamma\in \tilde{E}\label{margine11}. 
\end{align}
These functions were chosen ``wider'' than the $g_\gamma$'s, and obey 
\begin{align}\label{extras1}
\tilde{g}_\gamma g_\gamma=g_\gamma,\quad {\rm dist}\{{\rm
  supp}(D\tilde{g}_{\gamma}),{\rm supp}(g_{\gamma})\}\sim
L^\alpha,\qquad \gamma\in\tilde{E}.
\end{align}
By $D\tilde{g}_{\gamma}$ we mean that we take at least one derivative
of $\tilde{g}_{\gamma}$. 

Now let us define the third type of cut-offs,
$\{\tilde{\tilde{g}}_{\gamma}\}_{\gamma\in E}$:
\begin{align}
&{\rm supp}(\tilde{\tilde{g}}_{\gamma})\subset K(\gamma, 12 L^\alpha), \gamma\in
E;\quad \tilde{g}_{\gamma}(\x)=1 \: {\rm
  if}\: \x\in K(\gamma, 11 L^\alpha)\label{margine91}\\
&||D^\beta \tilde{\tilde{g}}_\gamma||_\infty\sim L^{-\alpha |\beta|},\quad \forall 
\beta\in \mathbb{N}^3,\quad  \forall \gamma\in E\label{margine111}. 
\end{align}
Note that we have $\tilde{\tilde{g}}_\gamma
\tilde{g}_\gamma=\tilde{g}_\gamma$ (the origin is not
considered here). 

\subsection{Proof of \eqref{efel22}} 

Define for every $\gamma\in E$:
\begin{align}\label{geigi}
{\bf A}_{\gamma}(\x):=\tilde{\tilde{g}}_\gamma(\x)\;{\bf A}(\x-\gamma).
\end{align}
Due to the support properties of our cut-off functions, we have the
estimates
\begin{align}\label{geigi2}
&\tilde{g}_\gamma\; {\bf A}_{\gamma}=\tilde{g}_\gamma\;{\bf
  A}(\x-\gamma),\nonumber \\
&||{\bf A}_{\gamma}||_{C^1(\R^3)}\leq \C\cdot L^\alpha.
\end{align}
Define for every $\gamma\in E$ (see also \eqref{feshch} and \eqref{impuls}):
\begin{align}\label{geigi3}
\mathbf{P}_\gamma(B):=\mathbf{P}(0)+B{\bf A}_{\gamma},\: 
H_L(B,\gamma):=\mathbf{P}_\gamma(B)^2+V,
\end{align}
where the Hamiltonian is defined with Dirichlet boundary conditions. 
Note that $H_L(B,\gamma)-H_L(0)$
is a relatively bounded perturbation of $H_L(0)$. 

Define the operator 
\begin{align}\label{geigi4}
U_L(B,z)& :=\tilde{g}_0 (H_\infty(B)-z)^{-1}g_0\nonumber \\
&+\sum_{\gamma\in E}
e^{iB\phi_0(\cdot,\gamma)}\tilde{g}_\gamma (H_L(B,\gamma)-z)^{-1}
e^{-iB\phi_0(\cdot,\gamma)}g_\gamma .
\end{align}
One can prove that the range of $U_L(B,z)$ is in the domain of
$H_L(B)$ and we have:
\begin{align}\label{geigi5}
&(H_L(B)-z)U_L(B,z)=1 +V_L(B,z),\\
&V_L(B,z):=\{-2i(\nabla\tilde{g}_0)\cdot \mathbf{P}(B)-(\Delta\tilde{g}_0)\} 
(H_\infty(B)-z)^{-1}g_0\nonumber \\
&+\sum_{\gamma\in E}
e^{iB\phi_0(\cdot,\gamma)}\{-2i(\nabla\tilde{g}_\gamma)\cdot
\mathbf{P}_\gamma(B)-(\Delta\tilde{g}_\gamma)\} (H_L(B,\gamma)-z)^{-1}
e^{-iB\phi_0(\cdot,\gamma)}g_\gamma .\nonumber 
\end{align}
In order to obtain this equality we used the locality of our
operators and various support properties of our cut-off functions, the
identity \eqref{fazam2}, definition \eqref{geigi}, and
\eqref{margine6}. 

Then we can write  
\begin{align}\label{geigi6}
(H_L(B)-z)^{-1}=U_L(B,z)+(H_L(B)-z)^{-1}V_L(B,z).
\end{align}
The good thing about $V_L(B,z)$ is that its operator 
norm is exponentially small in $L^\alpha$. This is
because we have the boundedness from \eqref{expdek} and
\eqref{expdek2} (valid also for $H_L(B,\gamma)$, as can easily be seen
from the proofs), and because of the estimate in
\eqref{extras1}. Indeed, for terms involving $\gamma\neq 0$, 
put $\x_0=\gamma$ in the two exponential
estimates, and take $+\delta$ on the left and $-\delta$ on the
right. Then we gain an overall decaying term from the left as \eqref{extras1} implies:
\begin{align}\label{geigi7}
\sup_{\x\in {\rm supp}(D\tilde{g}_\gamma) }\sup_{\x'\in {\rm
    supp}(g_\gamma) }e^{-\frac{\delta}{\langle r\rangle}(\langle
    \x-\gamma\rangle -\langle
    \x'-\gamma\rangle) } \leq e^{- \frac{\delta_1}{\langle
    r\rangle}L^\alpha},\; \gamma\neq 0,
\end{align}
where $\delta_1>0$ is small enough and $L$ is larger than some
$L_0$. 

For $\gamma=0$ the situation is slightly different, because we did not
assume convexity for $\Lambda_L$. But one of the terms whose norm we need to
estimate is (see \eqref{geigi5}) 
$$(\nabla\tilde{g}_0)\cdot \mathbf{P}(B) (H_\infty(B)-z)^{-1}g_0.$$
From \eqref{intkernrez} follows that the integral kernel of this
operator is bounded by 
\begin{align}\label{ggeigi8}
|(\nabla\tilde{g}_0)\cdot \mathbf{P}(B)
 (H_\infty(B)-z)^{-1}g_0|(\x,\x')\leq \C(\eta)\langle r\rangle ^M
 e^{-\frac{\delta}{\langle r\rangle }|\x-\x'|}.
\end{align}
Because $\x$ and $\x'$ are always separated by $\sim L^\alpha$ (see
the support properties for our cut-offs), we can write:
\begin{align}\label{ggeigi9}
|(\nabla\tilde{g}_0)\cdot \mathbf{P}(B)
 (H_\infty(B)-z)^{-1}g_0|(\x,\x')\leq \C(\eta)\langle r\rangle ^M
 e^{-\frac{\delta_1}{\langle r\rangle }|\x-\x'|}
e^{-\frac{\delta_2}{\langle r\rangle }L^\alpha},
\end{align}
where $\delta_1$ and $\delta_2$ are smaller than $\delta$. 

Therefore we can write for all $N\geq 1$:
\begin{align}\label{geigi8}
||V_L(B,z)|| &\leq \C(\eta)\cdot  L^{2-2\alpha} \langle r\rangle ^M
e^{- \frac{\delta_3}{\langle r\rangle}L^\alpha}\leq \C
(\eta,\alpha,N)\cdot  L^{-N} \langle r\rangle ^{M_1},
\end{align}
where we have to remember that we have $\sim L^{2-2\alpha}$ of
$g_\gamma$'s. The second estimate says that the norm decays faster
than any power of $L$, with the price of a higher power in $\langle
r\rangle$.

We now want to show that $V_L(B,z)$ does not contribute to the
thermodynamic limit of $\sigma_{L}(B)$. We have the following result:
\begin{proposition}\label{prop22}
\begin{align}\label{kompar1}
&  \sigma_{L}(B)= -\frac{1}{ {\rm Vol}(\Lambda_L)} 
\cdot {\rm Tr}
\int_{\Gamma_\omega}
{f}_{FD}(z) 
\left \{ P_1(B)U_L(B,z)P_2(B) 
U_L(B,z+\omega)\right .\nonumber \\
&\left . +  P_1(B)U_L(B,z-\omega)P_2(B) 
U_L(B,z)\right \}dz +\mathcal{O}(L^{-\infty}).
\end{align}
\end{proposition}

\noindent {\bf Proof}. 
The main idea is to show that when we replace $(H_L(B)-z)^{-1}$
in $F_L$ with the right hand side of \eqref{geigi6}, all terms containing
$V_L(B)$ will generate (after integrating by parts with respect to
$z$) some operators which in the trace norm will decay faster than any
power of $L$. 

If the differentiation does not act on $V_L(B)$ but on the other
resolvents, then it is enough to know that in the norm of
$B(L^2)$ it goes to zero faster than any power in $L$, as we saw in
\eqref{geigi8}. 

If the differentiation with respect to $z$ acts on $V_L(B)$, then
there will be a few terms which must be separately considered, and
prove their smallness in the trace norm. 

To give an example, after differentiating $N-1$ 
times with respect to $z$ ($N$ large), we obtain a term containing a
factor like:
\begin{equation}\label{buku01}
\sum_{j=1}^3(\partial_j\tilde{g}_0) P_j(B) (H_\infty(B)-z)^{-N}g_0.
\end{equation}
We will prove (in the trace norm) that it decays faster than any power
of $L$, times some polynomially bounded factor in $\langle
r\rangle$. Note that all the other factors multiplying the above operator are bounded 
operators, with a norm which is polynomially bounded in $\langle r\rangle$. 

Let us start with a technical result:
\begin{lemma}\label{bucurestilema1}
Let $Q_1$ and $Q_2$ be 
two compact unit cubes such that ${\rm dist}(Q_1,Q_2)=d >1$, and let $\chi_1$, $\chi_2$ denote 
their characteristic functions. Let $\alpha\in \{0,1\}$ and $j\in\{1,2,3\}$. Then if $N$ is large enough, 
there exist three 
constants $\delta_2 >0$, $N_1>1$ and $C>0$, all three independent of $z\in \Gamma_\omega$, $d$, $\alpha$, $j$ 
and $Q$'s such that 
\begin{align}\label{buku1}
||\chi_1 P_j^\alpha(B) (H_\infty(B)-z)^{-N} \chi_2||_{B_1}\leq C r^{N_1}\exp\{-d\delta_2/r\}. 
\end{align}
\end{lemma}

\begin{proof}
We assume that $\alpha=1$, the other case being similar. The strategy is to write our operator as a product of two 
Hilbert-Schmidt operators. By commuting $P_j(B)$ with one resolvent, we can
rewrite our operator as:
$$\chi_1 (H_\infty(B)-z)^{-1}T_j
(H_\infty(B)-z)^{-N+1}\chi_2,$$
where $T_j$ is a bounded operator which contains factors like $P_k(B)(H_\infty(B)-z)^{-1}$. 

Denote by $\x_{2}$ an arbitrary point in the support of $\chi_{2}$. 
We insert some exponentials in the following way:
\begin{align}\label{ggg1}
&{\chi}_1 
e^{- \frac{\delta_2}{\langle r\rangle}\langle
  \cdot -\x_{2}\rangle} \{\chi_{1}e^{
  \frac{\delta_2}{\langle r\rangle}\langle \cdot -\x_{2}\rangle } 
(H_\infty(B)-z)^{-1}e^{ -\frac{\delta_2}{\langle
    r\rangle}\langle
  \cdot -\x_{2}\rangle }\}\\ \nonumber 
&\cdot \{e^{
  \frac{\delta_2}{\langle r\rangle}\langle
  \cdot  -\x_{2}\rangle }T_j(H_\infty(B)-z)^{-N+1}e^{
  -\frac{\delta_2}{\langle r\rangle}\langle
  \cdot  -\x_{2}\rangle }\chi_{2}\}e^{
  \frac{\delta_2}{\langle r\rangle}\langle
  \cdot  -\x_{2}\rangle }\chi_{2}.
\end{align}
 If $\delta_2<\delta$ we can write:
$$ e^{-\frac{\delta_2}{\langle r\rangle}\langle
  \x  -\x_{2}\rangle} e^{
  \frac{-\delta}{\langle r\rangle}\langle
  \x  -\x'\rangle }e^{
  \frac{\delta_2}{\langle r\rangle}\langle
  \x'  -\x_{2}\rangle }\leq {\rm const} \: e^{-
  \frac{\delta-\delta_2}{\langle r\rangle}|\x  -\x'| }$$
Now the two factors containing resolvents in \eqref{ggg1} are Hilbert-Schmidt due to the 
presence of the cut-offs $\chi$ and the exponential decay of our kernels  
(Proposition \ref{prop2} and \eqref{intkernrez}).
Their 
Hilbert-Schmidt norm will grow polynomially with $r$, but be independent of $d$. 
The choice of $\x_2$ provides the decaying exponential factor on the right hand side of \eqref{buku1}. \end{proof}

Let us go back to \eqref{buku01} and try to use the previous lemma. We will show that in the trace norm, 
this operator decays exponentially in $L^\alpha$. Consider only one $j$. Cover both ${\rm supp}\{ \partial_j
\tilde{g}_0\}$ and  ${\rm supp}\{ g_0\}$ with disjoint cubes centered at points in $\mathbb{Z}^3$ and side length
equal to $1$; we only use cubes which have common points with the respective supports. We then have 
\begin{align}\label{oggg1}
&(\partial_j\tilde{g}_0) P_j(B)
(H_\infty(B)-z)^{-N}g_0 \nonumber \\
&=\sum_{s,s'}\tilde{\chi}_s(\partial_j\tilde{g}_0)  P_j(B)
(H_\infty(B)-z)^{-N}g_0\chi_{s'},
\end{align}
where $\tilde{\chi}_s$ and respectively $\chi_s$ denote the characteristic 
function of such unit
cubes which cover ${\rm supp}\{ \partial_j
\tilde{g}_0\}$ and respectively ${\rm supp}\{ g_0\}$. 
The number of cubes needed to cover the support of $g_0$ is of order
$L^3$, while for the other one is of order $L^{3\alpha}$; 
hence we have about $L^{3+3\alpha}$ terms in the above double sum. 

But each operator of the form 
$$\tilde{\chi}_s(\partial_j\tilde{g}_0) P_j(B)
(H_\infty(B)-z)^{-N}g_0\chi_{s'}$$
is exponentially small in the trace norm due to Lemma \ref{bucurestilema1}, since the distance between any two 
supports of $\tilde{\chi}_s$ and $\chi_{s'}$ is of order $L^\alpha$. Hence the entire sum in \eqref{oggg1} will
be ($r$ dependent) exponentially small in $L^\alpha$. But then we can trade off the fading exponential decay with a 
polynomial decay in $L$ and a polynomial growth in $r$ as we did in \eqref{geigi8}. So this term is under control.

Now let us go back to the beginning of the proof of Proposition \ref{prop22}. 
Other ``bad'' terms from the remainder in \eqref{kompar1} after
differentiation with respect to $z$ will contain powers
of $(H_L(B,\gamma)-z)^{-1}$, like for example
\begin{equation}\label{djh1}
(\partial_j\tilde{g}_\gamma)
P_{j,\gamma}(B)(H_L(B,\gamma)-z)^{-N}g_\gamma .
\end{equation}

 Here we
cannot easily commute with $P$'s due to various boundary terms. But we
do not need to do that. Look at \eqref{expdek4}, where we put 
$s=\frac{\delta}{\langle r\rangle }$ and $\x_0=\gamma$. 
Because ${\bf A}_\gamma$ is bounded from above by $L^\alpha$ (see
\eqref{geigi2}), it follows:
\begin{equation}\label{fff2}
||(-\Delta_D+1)(H_L(B,\gamma)-z)^{-1}||\leq \langle r\rangle ^M
L^\alpha.
\end{equation}
It means that for small $\delta$, 
the resolvent $e^{
  \frac{\delta}{\langle r\rangle}\langle
  \cdot  -\gamma \rangle }(H_L(B,\gamma)-z)^{-1}e^{-
  \frac{\delta}{\langle r\rangle}\langle
  \cdot  -\gamma \rangle }$ sandwiched with exponentials remains 
Hilbert-Schmidt
(with a norm which does not grow faster than  the $B_2$ 
norm of $(-\Delta_D+1)^{-1}$ times $L^\alpha$ and some polynomial in
$\langle r\rangle$). Since we have $N$ resolvents, the product of two
of them will give a trace class operator. Now we can repeat the
insertion of exponentials as we did in \eqref{ggg1}, and use
\eqref{geigi7} for getting the exponential decay in $L^\alpha$. 

 \qed 

\vspace{0.5cm}

Now let us show that all terms involving the sum over $\gamma\in E$ in
\eqref{geigi4} will not contribute at the end. The explanation is that
these terms are ''localized near boundary''. We can formulate the result
as follows:

\begin{proposition}\label{prop23} If $\alpha >0$ is small enough, then:
\begin{align}\label{aproxunu}
&  \lim_{L\to \infty} \left \{ 
\sigma_{L}(B) \frac{}{} +\right . \nonumber \\
& \left . \frac{1}{ {\rm Vol}(\Lambda_L)} 
\cdot {\rm Tr}
\int_{\Gamma_\omega}
{f}_{FD}(z) 
\left [ P_1(B)\tilde{g}_0 (H_\infty(B)-z)^{-1}g_0 
P_2(B) \tilde{g}_0 (H_\infty(B)-z-\omega)^{-1}g_0 \frac{}{} \right . \right
.\nonumber \\
&\left . \left . +P_1(B)\tilde{g}_0 (H_\infty(B)-z+\omega)^{-1}g_0 P_2(B) \tilde{g}_0
(H_\infty(B)-z)^{-1}g_0
 \frac{}{} \right ]dz \right \}=0.
\end{align}
\end{proposition}

\noindent {\bf Proof}. If we compare this with \eqref{geigi4},
we see that we need to show that all terms containing factors localized near boundary will converge to zero. 
Let us look at one such term, and prove the next lemma:
\begin{lemma}\label{lema1} Assume that $0<\alpha<1/3$. Then we have that
\begin{align}\label{kompar2}
&\lim_{L\to \infty}\frac{1}{ {\rm Vol}(\Lambda_L)} 
\cdot {\rm Tr}
\int_{\Gamma_\omega}
{f}_{FD}(z) 
 P_1(B)\tilde{g}_0 (H_\infty(B)-z)^{-1}g_0  \nonumber \\
&\cdot   P_2(B)\sum_{\gamma\in E}
e^{iB\phi_0(\cdot,\gamma)}\tilde{g}_\gamma (H_L(B,\gamma)-z-\omega)^{-1}
e^{-iB\phi_0(\cdot,\gamma)}g_\gamma dz =0.
\end{align}
\end{lemma}

\noindent {\bf Proof.} 
Using (\ref{fazam2}), and the fact that $\tilde{\tilde{g}}_\gamma
\tilde{g}_\gamma=\tilde{g}_\gamma$, we can rewrite the above term as
\begin{align}\label{kompar3}
&\frac{1}{ {\rm Vol}(\Lambda_L)} 
\cdot 
\sum_{\gamma\in E}{\rm Tr}\int_{\Gamma_\omega}
{f}_{FD}(z) 
 P_1(B)\tilde{g}_0 (H_\infty(B)-z)^{-1}g_0 \tilde{\tilde{g}}_\gamma 
\nonumber \\
&\cdot  
e^{iB\phi_0(\cdot,\gamma)}P_{2,\gamma}(B)\tilde{g}_\gamma 
(H_L(B,\gamma)-z-\omega)^{-1}
e^{-iB\phi_0(\cdot,\gamma)}g_\gamma dz.
\end{align}
After integrating by parts $N-1$ times with respect to $z$, 
we have to deal with
several situations. Let us take one resulting term (just the operator in the 
integrand):

\begin{align}\label{kompar4}
& P_1(B)\tilde{g}_0 (H_\infty(B)-z)^{-1}g_0 \tilde{\tilde{g}}_\gamma  
e^{iB\phi_0(\cdot,\gamma)}P_{2,\gamma}(B)\tilde{g}_\gamma 
(H_L(B,\gamma)-z)^{-N}
e^{-iB\phi_0(\cdot,\gamma)}g_\gamma ,
\end{align}
and let us estimate its trace norm. The factor containing
$H_\infty(B)$ is just bounded, and we cannot use it as a
Hilbert-Schmidt factor. What we do is to commute $\tilde{g}_\gamma$
over one resolvent to the right and get the identity:
\begin{align}\label{kompar5}
&P_{2,\gamma}(B)\tilde{g}_\gamma (H_L(B,\gamma)-z-\omega )^{-N}
g_\gamma\\
& =P_{2,\gamma}(B)(H_L(B,\gamma)-z-\omega)^{-1}\tilde{g}_\gamma
(H_L(B,\gamma)-z)^{-N+1}
g_\gamma\nonumber \\
&+P_{2,\gamma}(B)(H_L(B,\gamma)-z-\omega)^{-1}[H_L(B,\gamma),\tilde{g}_\gamma]
(H_L(B,\gamma)-z-\omega)^{-N}g_\gamma.\nonumber 
\end{align}
Now the second term will again contain at least one derivative of
$\tilde{g}_\gamma$ and reasoning as we did for \eqref{djh1} we can
show it will be exponentially small. Let
us look at the first term. The operator
$P_{2,\gamma}(B)(H_L(B,\gamma)-\zeta)^{-1}$ is only bounded. 
But $\tilde{g}_\gamma
(H_L(B,\gamma)-\zeta)^{-N+1}
g_\gamma$ is trace class and 
\begin{align}\label{pertii}
||\tilde{g}_\gamma
&(H_L(B,\gamma)-\zeta)^{-N+1}
g_\gamma||_{B_1}\\ 
&\leq \frac{1}{\eta^{N-2}}||\tilde{g}_\gamma
(H_L(B,\gamma)-\zeta)^{-1}||_{B_2}\cdot ||
(H_L(B,\gamma)-\zeta)^{-1}g_\gamma||_{B_2}.\nonumber 
\end{align}
Hence using \eqref{fff2} we have 
$$||\tilde{g}_\gamma 
(H_L(B,\gamma)-\zeta)^{-1}||_{B_2}\leq \C \langle r\rangle ^M
||\tilde{g}_\gamma (-\Delta_D+1)^{-1}||_{B_2} \cdot L^\alpha.$$
But the Hilbert-Schmidt norm of $\tilde{g}_\gamma (-\Delta_D+1)^{-1}$
is of order of the square root of the support of $\tilde{g}_\gamma$,
that is $L^{3\alpha/2}$ (use here \eqref{intkern} with $B=0$). 
The other factor comes with a similar contribution, hence we can write 
$$||\tilde{g}_\gamma
(H_L(B,\gamma)-z-\omega )^{-N+1}
g_\gamma||_{B_1}\leq \C\: \langle r\rangle ^M \cdot L^{5\alpha}.$$
Now using this in the integral with respect to $z$, this particular
term will give a contribution of  $L^{5\alpha}$ for each
$\gamma$. Since we have $\sim L^{2-2\alpha}$ different $\gamma$'s, the
total contribution will be bounded by $L^{2+3\alpha}$. But if $\alpha
<1/3$, after we divide with the volume of $\Lambda_L$ it will converge
to zero. 

Now let us go back to \eqref{kompar3}, and see that after integration
by parts we can get a term like 
\begin{align}\label{kompar6}
& P_1(B)\tilde{g}_0 (H_\infty(B)-z)^{-N}g_0 
\tilde{\tilde{g}}_\gamma \nonumber \\
&\cdot  
e^{iB\phi_0(\cdot,\gamma)}P_{2,\gamma}(B)\tilde{g}_\gamma
(H_L(B,\gamma)-z-\omega )^{-1}g_\gamma .
\end{align}
Here the operator $P_{2,\gamma}(B)\tilde{g}_\gamma
(H_L(B,\gamma)-z)^{-1}$ is not Hilbert-Schmidt, so we have to look at
the first factor. 

We can write 
\begin{align}\label{kompar7}
&P_1(B)\tilde{g}_0 (H_\infty(B)-z)^{-N}g_0 \tilde{\tilde{g}}_\gamma\\
&=P_1(B)\tilde{g}_0 (H_\infty(B)-z)^{-N+1}\tilde{\tilde{g}}_\gamma
(H_\infty(B)-z)^{-1}g_0\nonumber \\
&+ P_1(B)\tilde{g}_0 (H_\infty(B)-z)^{-N}
[H_\infty(B),\tilde{\tilde{g}}_\gamma](H_\infty(B)-z)^{-1}g_0.\nonumber 
\end{align}
In the first term, the function  $\tilde{\tilde{g}}_\gamma$ makes the
two resolvents next to it become Hilbert-Schmidt, each having a norm
proportional with $L^{3\alpha/2}$ and some power of $\langle r\rangle$
(use the exponential decay of the kernels). 
So this term is ''good'', considering that we have to divide with $L^3$ in the end. 
The second term contains at least one
derivative of $\tilde{\tilde{g}}_\gamma$ (here 
$[H_\infty(B),\tilde{\tilde{g}}_\gamma]$ is linear in $P_j(B)$'s), 
together with factors like 
$$\{P_1(B)\tilde{g}_0 (H_\infty(B)-z)^{-N}P_j(B)\}
(\partial_j\tilde{\tilde{g}}_\gamma)(H_\infty(B)-z)^{-1}g_0.$$
Now here we can use the estimate from Proposition \ref{prop6} and see
that we again have a product of two Hilbert-Schmidt operators: the first Hilbert-Schmidt norm will be proportional 
with $L^{3/2}$, while the other one will behave like  $L^{3\alpha/2}$.  When we divide by $L^3$, this contribution will go to zero.

All the other terms resulting from integrating by parts with respect
to $z$ can be treated in a similar way. \qed

\vspace{0.5cm}

Now let us go back to \eqref{kompar1} and analyze another boundary term:
\begin{lemma}\label{lema2} For every $0<\alpha<1/3$ we have:
\begin{align}\label{kompar8}
&\lim_{L\to \infty}\frac{1}{ {\rm Vol}(\Lambda_L)} 
\cdot {\rm Tr}
\int_{\Gamma_\omega}
{f}_{FD}(z) \nonumber \\
& \cdot 
 P_1(B)  \sum_{\gamma'\in E}
e^{iB\phi_0(\cdot,\gamma')}\tilde{g}_{\gamma'} (H_L(B,\gamma')-z)^{-1}
e^{-iB\phi_0(\cdot,\gamma')}g_{\gamma'}\nonumber \\
&\cdot   P_2(B)\sum_{\gamma\in E}
e^{iB\phi_0(\cdot,\gamma)}\tilde{g}_\gamma (H_L(B,\gamma)-z-\omega )^{-1}
e^{-iB\phi_0(\cdot,\gamma)}g_\gamma dz =0.
\end{align}
\end{lemma}

\noindent{\bf Proof.} Let us note that when keeping $\gamma$ fixed, only a finite number ($L$-independent) 
of $\gamma'$'s will have an 
overlapping support. This means that the above double sum will only contain around $L^{2-2\alpha}$ non-zero terms. 
Now use again \eqref{fazam} and integration by parts
with respect to $z$. Each non-zero term in the double sum will be a product of two Hilbert-Schmidt operators, each 
with a Hilbert-Schmidt norm of the order of $L^{5\alpha/2}$. The total trace norm will grow at most like 
$L^{2+3\alpha}$, hence if $\alpha<1/3$ this term will not contribute. We do not give more
details. \qed 

\vspace{0.5cm}

The last ingredient in proving \eqref{efel22} is contained in the following result (see also \eqref{aproxunu}):
\begin{proposition}\label{prop222}
\begin{align}\label{limtherm1}
&  \lim_{L\to \infty}-\frac{1}{ {\rm Vol}(\Lambda_L)} 
\cdot {\rm Tr}
\int_{\Gamma_\omega}
{f}_{FD}(z) 
\left \{ P_1(B)\tilde{g_0}(H_\infty(B)-z)^{-1}g_0 P_2(B) \tilde{g_0}(H_\infty(B)-z-\omega)^{-1}g_0\frac{}{}\right .\nonumber \\
&\left .\frac{}{} P_1(B)\tilde{g_0}(H_\infty(B)-z+\omega)^{-1}g_0 P_2(B) \tilde{g_0}(H_\infty(B)-z)^{-1}g_0
 \right \}dz =-\int_\Omega s_B(\x)d\x.
\end{align}
\end{proposition}

\noindent{\bf Proof.} First, due to support properties, we have
\begin{equation}\label{intsupr}
g_0P_2(B) \tilde{g_0}=g_0P_2(B)=P_2(B)-(1-g_0)P_2(B). 
\end{equation}
Because $1-g_0$ is supported outside of a thin region around the boundary of
$\Lambda_L$, all terms generated by $(1-g_0)P_2(B)$ will converge to zero; let us
prove this. After integrating $N-1$ by parts with respect to $z$, we obtain
several terms from the integrand which look like this one:
($N_1+N_2\geq N\geq 10 $):
\begin{equation}\label{intsupr2}
 P_1(B)\tilde{g_0}(H_\infty(B)-z)^{-N_1}(1-g_0)P_2(B)
(H_\infty(B)-z-\omega)^{-N_2}g_0.
\end{equation}
Due to the symmetry of this term, assume without loss of generality that $N_1\geq 5$. Then by writing
$1-g_0=(1-g_0)\chi_{\Lambda_L}+(1-g_0)(1-\chi_{\Lambda_L})$, we get two 
types of
contributions. The one coming from 
$(1-g_0)(1-\chi_{\Lambda_L})=(1-\chi_{\Lambda_L})$ is
localized outside $\Lambda_L$, and its trace norm will be
exponentially small in $L^\alpha$; we can apply 
Lemma \ref{bucurestilema1} since the distance between $\Lambda_L^c$ and the
supports of $\tilde{g}_0$ and $g_0$ is of order $L^\alpha$. The number of needed 
covering unit cubes is polynomially bounded in 
$L$.

Thus the only contribution from \eqref{intsupr2} can come from:
\begin{equation}\label{intsupr3}
 P_1(B)\tilde{g_0}(H_\infty(B)-z)^{-N_1}(1-g_0)\chi_{\Lambda_L}
P_2(B)(H_\infty(B)-z-\omega)^{-N_2}g_0.
\end{equation}
Here we can write:
\begin{align}\label{intsupr4}
 & P_1(B)\tilde{g_0}(H_\infty(B)-z)^{-N_1}(1-g_0)\chi_{\Lambda_L}\\
& \{P_1(B)\tilde{g_0}(H_\infty(B)-z)^{-2}\}\cdot
 \{(H_\infty(B)-z)^{-N_1+2}(1-g_0)\chi_{\Lambda_L}\},\nonumber 
\end{align}
where both factors are Hilbert-Schmidt, with kernels exponentially
localized near diagonal as in Propositions \ref{prop5} and
\ref{prop6}. The Hilbert-Schmidt norm of the first factor is bounded
by $L^{3/2}$, while for the second one we have a bound of
$L^{1+\alpha/2}$ (square roots of certain volumes). The trace norm of
the product is thus bounded by $L^{5/2+\alpha/2}$ and some polynomial
in $\langle r \rangle$. After integrating with respect to $z$, and
dividing with $L^3$ (the volume of $\Lambda_L$), this term will
converge to zero provided $\alpha <1$. 

Now we can go back to \eqref{intsupr} and analyze the term generated
by $P_2(B)$. Because there are no other cut-offs in the middle, and
because the commutator $[P_1(B),\tilde{g}_0]$ will generate another fast
decaying term, we see
that we have just proved the following identity (see Theorem
\ref{teorema1} ii for the definition of $F_\infty$):
 \begin{align}\label{intsupr5}
 \lim_{L\to \infty}\left (\sigma_L(B) +\frac{1}{{\rm
       Vol}(\Lambda_L)}{\rm Tr}\{\tilde{g}_0 F_\infty g_0\}\right ).
\end{align}
But the operator $\tilde{g}_0 F_\infty g_0$ is trace class, with a
jointly continuous kernel, hence (see \eqref{margine1}, 
\eqref{margine4} and \eqref{margine8})
$${\rm Tr}\{\tilde{g}_0 F_\infty g_0\}=\int_{{\rm supp}(g_0)}
s_B(\x,\x) g_0(\x)d\x.$$
Using the periodicity of $s_B$ with respect to $\mathbb{Z}^3$ and the 
support properties of $g_0$, we
finally get:
\begin{align}\label{intsupr6}
 \lim_{L\to \infty}\sigma_L(B)=-\int_{\Omega}
s_B(\x,\x)d\x
\end{align}
and the proof of \eqref{efel22} is over. 

\qed

\subsection{Proof of \eqref{efel2}.} 

We start by investigating $\partial_B\sigma_L(0)$ and
try to put it in a form which is better suited for the thermodynamic
limit. 

Note that we have already argued that $\sigma_L$ was smooth in $B$
near zero (see the remark around \eqref{derivrez}). The hard part is
to show that the polynomial growth in $L$ of the trace norm does not appear
in the actual trace. Similar difficulties involving magnetic semi-groups were 
encountered in \cite{ABN, Cor, BCL1, BCL2}.

\subsubsection{Only $U_L(B)$ counts.}

\noindent{\bf Notation}. In order to shorten our formulas, we will
sometimes write $(z\to z-\omega)$, which means that we have to repeat the
previously appearing product of operators in which we have to change
$z$ with $z-\omega$. 

\vspace{0.5cm}

First, let us prove that even if we differentiate with respect to $B$,
we still have a result similar to Proposition \ref{prop22}. 
\begin{proposition}\label{lemma100} At $B=0$ we have:
\begin{align}\label{derkompar1}
&  \partial_B \left \{\sigma_{L}(B) +\frac{1}{ {\rm Vol}(\Lambda_L)} 
\cdot {\rm Tr}
\int_{\Gamma_\omega}
{f}_{FD}(z) 
\left \{ P_1(B)U_L(B,z)P_2(B) 
U_L(B,z+\omega)\right .\right . \nonumber \\
&+ \left . \left . z \rightarrow z-\omega \right \}dz 
\frac{}{} \right \}=\mathcal{O}(L^{-\infty}).
\end{align}
\end{proposition}

\noindent{\bf Proof}. As usual, before differentiating with respect to
$B$ one has to perform a certain number of integrations by parts with
respect to $z$. There will appear many terms in the remainder
containing $V_L(B,z)$ (see \eqref{geigi5}), but all of them will have
the distinctive feature of containing pairs of cut-off
functions whose supports are at a distance $\sim L^\alpha$ one from
another. 

We do not want to treat all the situations which can appear here, and
instead we will only prove a typical estimate needed for the result.

\begin{lemma}\label{lemma223}
Assume that $N$ is large enough, and choose $z$ such that ${\rm
  dist}\{z,[0,\infty )\}=\eta>0$, $r=\Re(z)$. Then the map 
$$(-1,1)\ni B\to (\Delta \tilde{g}_0)(H_\infty(B)-z)^{-N}g_0 \in 
B_1(L^2)$$
is differentiable in the trace norm, and there exists $\delta_1$ small enough and $M$
large enough such that 
\begin{align}\label{derkompar2}
\left \Vert \{\partial_B(\Delta
    \tilde{g}_0)(H_\infty(B)-z)^{-N}g_0\}_{B=0}\right \Vert_{1} \leq
  \C(\eta,N)e^{-\frac{\delta_1}{\langle r\rangle }L^{\alpha}}
\langle r\rangle ^M.
\end{align}
Moreover, 
\begin{align}\label{derkompar20}
&\left \Vert (\Delta
    \tilde{g}_0)(H_\infty(B)-z)^{-N}g_0-(\Delta
    \tilde{g}_0)(H_\infty(0)-z)^{-N}g_0\right .  \nonumber \\
&\left . -B\{\partial_B(\Delta
    \tilde{g}_0)(H_\infty(B)-z)^{-N}g_0\}_{B=0}\right 
  \Vert_{B_1}\leq 
\C(\eta,N)e^{-\frac{\delta_1}{\langle r\rangle }L^{\alpha}}\langle r\rangle ^M.
\end{align}
\end{lemma}

\noindent{\bf Proof of the lemma}. We only prove differentiability at
$B=0$. The key ingredient is to give a proper sense to the formal identity:

\begin{align}\label{derkompar3}
&(H_\infty(B)-z)^{-1}-(H_\infty(0)-z)^{-1}\nonumber \\
&=-(H_\infty(B)-z)^{-1}\{H_\infty(B)-H_\infty(0)\}(H_\infty(0)-z)^{-1}.
\end{align}
As it stands, the right hand side makes no sense because
$H_\infty(B)-H_\infty(0)$ contains terms like $-2i B {\bf
  A}(\x)\cdot\nabla_\x$ which are not relatively bounded to the free
Laplacian due to the linear growth of our magnetic potential. For 
$0<\delta_2<\delta_3$ and 
$\x_0\in\R^3$, we can
write (still formally):
\begin{align}\label{derkompar4}
&e^{\frac{\delta_2}{\langle
    r\rangle }\langle \cdot -\x_0\rangle }
\{(H_\infty(B)-z)^{-1}-(H_\infty(0)-z)^{-1}\}e^{-\frac{\delta_3}{\langle
    r\rangle }\langle \cdot -\x_0\rangle } \\
&=-e^{\frac{\delta_2}{\langle
    r\rangle }\langle \cdot -\x_0\rangle }(H_\infty(B)-z)^{-1}
\{H_\infty(B)-H_\infty(0)\}(H_\infty(0)-z)^{-1}e^{-\frac{\delta_3}{\langle
    r\rangle }\langle \cdot -\x_0\rangle }.\nonumber
\end{align}
Note the fact that the growing exponential is weaker that the decaying
one. This identity now holds in the Hilbert-Schmidt class; in order to
see this, introduce the notation 
\begin{align}\label{derkompar6}
1<\alpha_1<\alpha_2<\frac{\delta_3}{\delta_2},\quad,
\end{align}
Then we can write the above identity as 
\begin{align}\label{derkompar5}
&e^{-\frac{\delta_3-\delta_2}{\langle
    r\rangle }\langle \cdot -\x_0\rangle }\left \{e^{\frac{\delta_3}{\langle
    r\rangle }\langle \cdot -\x_0\rangle }
\{(H_\infty(B)-z)^{-1}-(H_\infty(0)-z)^{-1}\}e^{-\frac{\delta_3}{\langle
    r\rangle }\langle \cdot -\x_0\rangle }\right \} \nonumber \\
&=-e^{-\frac{(\alpha_1-1)\delta_2}{\langle
    r\rangle }\langle \cdot -\x_0\rangle }\cdot 
\left \{e^{\frac{\alpha_1\delta_2}{\langle
    r\rangle }\langle \cdot -\x_0\rangle }
(H_\infty(B)-z)^{-1}e^{-\frac{\alpha_1\delta_2}{\langle
    r\rangle }\langle \cdot -\x_0\rangle }\right \}\nonumber \\
&\cdot \left \{ e^{\frac{\alpha_1 \delta_2}{\langle
    r\rangle }\langle \cdot -\x_0\rangle }
\{H_\infty(B)-H_\infty(0)\}e^{-\frac{\alpha_2\delta_2}{\langle
    r\rangle }\langle \cdot -\x_0\rangle }\right \}\nonumber \\
&\cdot \left \{e^{\frac{\alpha_2\delta_2}{\langle
    r\rangle }\langle \cdot -\x_0\rangle }
(H_\infty(0)-z)^{-1}e^{-\frac{\alpha_2\delta_2}{\langle
    r\rangle }\langle \cdot -\x_0\rangle }\right \}\cdot
 e^{-\frac{(\delta_3-\alpha_2\delta_2)}{\langle
    r\rangle }\langle \cdot -\x_0\rangle }.
\end{align}
Now $H_\infty(B)-H_\infty(0)=2B{\bf A}\cdot \mathbf{P}(0)+B^2{\bf
  A}^2$ is proportional with $B$, and the linear growth of ${\bf A}$ is
compensated by the higher exponential decay on the right hand side. 
Hence \eqref{expdek} and \eqref{expdek5} (at $B=0$) imply that if $\delta_2$ is small enough, we have:
\begin{align}\label{derkompar7}
\left \Vert e^{\frac{\alpha_1 \delta_2}{\langle
    r\rangle }\langle \cdot -\x_0\rangle }\{H_\infty(B)-H_\infty(0)\}
(H_\infty(0)-z)^{-1}e^{-\frac{\alpha_2\delta_2}{\langle
    r\rangle }\langle \cdot -\x_0\rangle }\right 
\Vert \leq \C\langle \x_0\rangle ^2 \cdot |B| \cdot \langle r\rangle^M.
\end{align}
The estimate \eqref{fazam7} (valid uniformly in $B\in [-1,1]$ and the
$\delta$ there should be chosen slightly larger than the $\delta_3$ )
tells us that: 
\begin{align}\label{derkompar77}
\left \Vert e^{-\frac{(\alpha_1-1)\delta_2}{\langle
    r\rangle }\langle \cdot -\x_0\rangle }\left 
\{e^{\frac{\alpha_1\delta_2}{\langle
    r\rangle }\langle \cdot -\x_0\rangle }
(H_\infty(B)-z)^{-1}e^{-\frac{\alpha_2\delta_2}{\langle
    r\rangle }\langle \cdot -\x_0\rangle }\right
\}\right\Vert_{B_2}\leq \C\; \langle r\rangle ^M.
\end{align}
 Hence we have proved the estimate 
\begin{align}\label{derkompar8}
&\left \Vert e^{\frac{\delta_2}{\langle
    r\rangle }\langle \cdot -\x_0\rangle }
\{(H_\infty(B)-z)^{-1}-(H_\infty(0)-z)^{-1}\}e^{-\frac{\delta_3}{\langle
    r\rangle }\langle \cdot -\x_0\rangle }\right \Vert_{B_2}
\leq \C |B|\langle \x_0\rangle ^2 \; \langle r\rangle^M.
\end{align}
Now go back to \eqref{derkompar4} and isolate the linear term in $B$:
\begin{align}\label{derkompar9}
&e^{\frac{\delta_2}{\langle
    r\rangle }\langle \cdot -\x_0\rangle }
\{(H_\infty(B)-z)^{-1}-(H_\infty(0)-z)^{-1}\}e^{-\frac{\delta_3}{\langle
    r\rangle }\langle \cdot -\x_0\rangle } \\
&=-Be^{\frac{\delta_2}{\langle
    r\rangle }\langle \cdot -\x_0\rangle }(H_\infty(0)-z)^{-1}\{2{\bf
  A}\cdot \mathbf{P}(0)\}
(H_\infty(0)-z)^{-1}e^{-\frac{\delta_3}{\langle
    r\rangle }\langle \cdot -\x_0\rangle }\nonumber\\
&-B^2e^{\frac{\delta_2}{\langle
    r\rangle }\langle \cdot -\x_0\rangle }
(H_\infty(0)-z)^{-1}\{{\bf A}^2\}
(H_\infty(0)-z)^{-1}e^{-\frac{\delta_3}{\langle
    r\rangle }\langle \cdot -\x_0\rangle }\nonumber\\
&-e^{\frac{\delta_2}{\langle
    r\rangle }\langle \cdot -\x_0\rangle }
\{(H_\infty(B)-z)^{-1}-(H_\infty(0)-z)^{-1}\}e^{-\frac{\tilde{\delta}}{\langle
    r\rangle }\langle \cdot -\x_0\rangle
}\nonumber \\
&\cdot e^{\frac{\tilde{\delta}}{\langle
    r\rangle }\langle \cdot -\x_0\rangle } 
\{H_\infty(B)-H_\infty(0)\}(H_\infty(0)-z)^{-1}e^{-\frac{\delta_3}{\langle
    r\rangle }\langle \cdot -\x_0\rangle },\nonumber 
\end{align}
where in the last term we inserted an exponential with
$\delta_2<\tilde{\delta}<\delta_3$. Now it is easy to get the
estimate:
\begin{align}\label{derkompar10}
&\left \Vert e^{\frac{\delta_2}{\langle
    r\rangle }\langle \cdot -\x_0\rangle }
\{(H_\infty(B)-z)^{-1}-(H_\infty(0)-z)^{-1}\}e^{-\frac{\delta_3}{\langle
    r\rangle }\langle \cdot -\x_0\rangle } \right . \nonumber \\
&\left . +Be^{\frac{\delta_2}{\langle
    r\rangle }\langle \cdot -\x_0\rangle }
(H_\infty(0)-z)^{-1}\{2{\bf
  A}\cdot \mathbf{P}(0)\}
(H_\infty(0)-z)^{-1}e^{-\frac{\delta_3}{\langle
    r\rangle }\langle \cdot -\x_0\rangle }\right
\Vert_{B_2}\nonumber \\
&= \C |B|^2 \langle \x_0\rangle ^4\; \langle r\rangle ^M.
\end{align}

This is enough to prove that if $N\geq 2$, the mapping in Lemma
\ref{lemma223} is differentiable in the trace norm sense at
$B=0$. Indeed, proceeding as we did for \eqref{oggg1}, we can insert
many cut-off functions and cover the supports of $g_0$ and
$\tilde{g_0}$. Using the same notations, we have that for some
$\delta_2<\delta_3<\delta_4$ we have 
\begin{align}\label{derkompar11}
&\tilde{\chi}_s(\Delta\tilde{g}_0)
(H_\infty(B)-z)^{-N}g_0\chi_{s'}\\
&=\tilde{\chi}_se^{-\frac{\delta_2}{\langle
    r\rangle }\langle \cdot -\x_{s'}\rangle }e^{\frac{\delta_2}{\langle
    r\rangle }\langle \cdot -\x_{s'}\rangle }(\Delta\tilde{g}_0) (H_\infty(B)-z)^{-N}g_0e^{-\frac{\delta_4}{\langle
    r\rangle }\langle \cdot -\x_{s'}\rangle }e^{\frac{\delta_4}{\langle
    r\rangle }\langle \cdot -\x_{s'}\rangle }\chi_{s'}.\nonumber
\end{align}
Now we can differentiate with respect to $B$ in the trace norm-sense,
using the result for the Hilbert-Schmidt norm and the fact that we
have at least two factors in $B_2$ (adapt
\eqref{derkompar77}). At last we again use the
fact that the distance between the supports of $\chi$'s is $\sim
L^\alpha$, and that all growing factors are just polynomials in $L$. We
consider Lemma \ref{lemma223} proved.\qed

\vspace{0.5cm}

Now we can use \eqref{derkompar20} in all the terms on the right hand
side of \eqref{derkompar1} which contain operators of the type treated
in Lemma \ref{lemma223}. The exponential decay of $f_{FD}$ can be used
to obtain a decay faster than any power of $L$. 

All other terms from the remainder can be treated in a similar manner,
and we consider Proposition \ref{lemma100} as proved. \qed

\vspace{0.5cm}

In the remaining part of our paper we will investigate the
thermodynamic limit of the main contribution to
$\partial_B\sigma_L(B)$, given by 
\begin{align}\label{laste1}
 &\partial_B {\rm Tr}
\int_{\Gamma_\omega}
{f}_{FD}(z) 
\left \{ P_1(B)U_L(B,z)P_2(B) U_L(B,z+\omega)+ ( z \rightarrow
  z-\omega) \right \} dz.
\end{align}

\subsubsection{The boundary terms}

Here the magnetic perturbation theory will play a crucial role. There
are several terms in the definition of $U_L$ (see \eqref{geigi4}), and
when we insert them into \eqref{laste1} they will generate even more
terms. 

We will now prove that only the term which contains two resolvents
with $H_\infty$ will contribute at the thermodynamic limit. The main
result can be stated in the following way:
\begin{proposition}\label{lemma300} At $B=0$ and $\alpha$ sufficiently small we have:
\begin{align}\label{derkompar009}
&  \lim_{L\to\infty}\partial_B \left \{\sigma_{L}(B) +\frac{1}{ {\rm Vol}(\Lambda_L)} 
\cdot {\rm Tr}
\int_{\Gamma_\omega}
{f}_{FD}(z) \right.\nonumber \\
&\cdot\left \{ P_1(B)\tilde{g}_0(H_\infty(B)-z)^{-1}g_0 P_2(B) 
\tilde{g}_0(H_\infty(B)-z-\omega)^{-1}g_0 \right . \nonumber \\
&+ \left . \left . z \rightarrow z-\omega \right \}dz \frac{}{} \right \}=0.
\end{align}
\end{proposition}

\noindent{\bf Proof}. Let us start with the following boundary term:

\begin{align}\label{laste10}
 &X_{1}(B,L):= {\rm Tr}
\int_{\Gamma_\omega}
{f}_{FD}(z) 
\left \{ \frac{}{}P_1(B)\tilde{g}_0(H_\infty(B)-z)^{-1}g_0 \right . \nonumber
\\ 
&\cdot P_2(B) \sum_{\gamma\in
    E}e^{iB\varphi_0(\cdot,\gamma)}\tilde{g}_\gamma
(H_L(B,\gamma)-z-\omega )^{-1}e^{-iB\varphi_0(\cdot,\gamma)}g_\gamma \nonumber \\
&\left . + ( z \rightarrow
  z-\omega)\frac{}{} \right \} dz.
\end{align}
We have already seen in Lemma \ref{lema1} that $X_{1}(B,L)$ divided by
the volume of $\Lambda_L$ converges to zero when $L$ converges to
infinity. Now we would like to show that $\{\partial_BX_{1}\}(0,L)$
has the same property. 

Proceeding as in \eqref{kompar3}, we can rewrite $X_1(B,L)$ as 
\begin{align}\label{laste11}
 &X_{1}(B,L)= \sum_{\gamma\in
    E} X_1(B,L,\gamma),\\
&\label{laste12}
 X_1(B,L,\gamma):={\rm Tr}
\int_{\Gamma_\omega}
{f}_{FD}(z) 
\left \{ \frac{}{}P_1(B)\tilde{g}_0(H_\infty(B)-z)^{-1}g_0 \right . \nonumber
\\ 
&\cdot e^{iB\varphi_0(\cdot,\gamma)}P_{2,\gamma}(B)\tilde{g}_\gamma
(H_L(B,\gamma)-z-\omega )^{-1}e^{-iB\varphi_0(\cdot,\gamma)}g_\gamma \nonumber \\
&\left . + ( z \rightarrow
  z-\omega)\frac{}{} \right \} dz.
\end{align}

We now prove the following estimate:
\begin{lemma}\label{lema30} For every $L\geq 1$  we have 
\begin{align}\label{laste14}
\sup_{\gamma\in E}\sup_{0<|B|\leq 1}\left \vert
  \frac{1}{B}\{X_1(B,L,\gamma)-X_1(0,L,\gamma)\}\right \vert \leq
\C(\alpha)\cdot 
L^{5\alpha} .
\end{align}
\end{lemma}

\noindent{\bf Remark.} Before starting the proof of this lemma, let us
note that it immediately implies that 
\begin{align}\label{laste13}
 \lim_{L\to\infty}\frac{1}{{\rm Vol}(\Lambda_L)}\{\partial_BX_{1}\}(0,L)= 0.
\end{align}
This is so because we know that $\{\partial_BX_{1}\}(0,L)$ exists, and
moreover, it 
must be bounded from above by the right hand side of \eqref{laste14} times
the number of $\gamma$'s in $E$, i.e. $\sim L^{2-2\alpha}$. Then if
$\alpha$ is chosen small enough, after dividing by $\sim L^3$ we get
something converging to zero. 

\noindent{\bf Proof of Lemma \ref{lema30}.} Define the function 
\begin{align}\label{laste15}
 & \tilde{ X}_1(B,L,\gamma):={\rm Tr}
\int_{\Gamma_\omega}
{f}_{FD}(z) 
\left \{ \frac{}{}P_1(B)\tilde{g}_0S_B(z)g_0 \right . \nonumber
\\ 
&\cdot e^{iB\varphi_0(\cdot,\gamma)}P_{2}(0)\tilde{g}_\gamma
(H_L(0)-z-\omega )^{-1}e^{-iB\varphi_0(\cdot,\gamma)}g_\gamma \nonumber \\
&\left . + ( z \rightarrow
  z-\omega)\frac{}{} \right \} dz.
\end{align}

The proof of this lemma has two parts. The first one will state that 
\begin{align}\label{ympos1}
 \sup_{\gamma\in E}\sup_{0<|B|\leq 1}\left \vert
  \frac{1}{B}\{X_1(B,L,\gamma)-\tilde{ X}_1(B,L,\gamma)\}\right \vert \leq
\C(\alpha)\cdot 
L^{5\alpha} ,
\end{align}
while the second one is 
\begin{align}\label{ympos2}
 \sup_{\gamma\in E}\sup_{0<|B|\leq 1}\left \vert
  \frac{1}{B}\{\tilde{X}_1(B,L,\gamma)-{ X}_1(0,L,\gamma)\}\right \vert \leq
\C(\alpha)\cdot 
L^{5\alpha}.
\end{align}

\noindent{\bf Part one}. 
The first estimate is not very much different from what we have
already done until now. One uses repeated integration by parts with
respect to $z$ in $X_1(B,L,\gamma)$ and then we expand each resolvent
$(H_\infty(B)-\zeta )^{-1}$ using the first equality in
\eqref{fazam6}, and each resolvent $(H_L(B,\gamma)-\zeta )^{-1}$
as perturbation of $(H_L(0)-\zeta )^{-1}$. Note that due to
\eqref{geigi2}, the growth induced by ${\bf A}_\gamma(\x)$ does not
exceed $L^\alpha$. As for the $H_\infty(B)$, commuting $P_j(B)$ with
the magnetic phases will always transform ${\bf
  A}(\x)$ into ${\bf
  A}(\x-\x')$, whose growth will now be tempered by the exponential
decay from \eqref{fazam5}. The general strategy for all terms arising from integration by
parts with respect to $z$ is to estimate the trace by the trace norm,
using the fact that the trace norm of a product is bounded by the
Hilbert-Schmidt norm of two factors. 
Several other estimates are needed in order to prove \eqref{ympos1}. Remember that $G_N(\x,\x';z)$ is the integral 
kernel of $(H_\infty(0)-z)^{-N}$, $N\geq 1$. Denote by $S^{(N)}_B(z)$ the operator
corresponding to the integral kernel
$e^{iB\varphi_0(\x,\x')}G_N(\x,\x';z)$. If $N=1$ they coincide with
the operators $S_B(z)$ defined in \eqref{fazam3}.Then we have: 

\noindent {\rm i}. For every $B\in [-1,1]$, and $N\geq 1$
\begin{align}\label{laste2}
||S^{(N)}_B(z)||\leq \C(N) \langle r\rangle ^M .
\end{align}

\noindent {\rm ii}. For every $B\in [-1,1]$, $N\geq 1$ and $Q\subset
\R^3$ a compact set:
\begin{align}\label{laste3}
||\chi_QS^{(N)}_B(z)||_{B_2}
\leq \C(N) \sqrt{{\rm Vol}(Q)}\;\langle r\rangle ^M.
\end{align}

\noindent {\rm iii}. For every $B\in [-1,1]$,
\begin{align}\label{laste4}
||(H_\infty(B)-z)^{-1}-S_B(z)+B S_B(z)T_B(z)||
\leq \C |B|^2\langle r\rangle ^M .
\end{align}

\noindent {\rm iv}. For every $B\in [-1,1]$, $N\geq 2$ and $Q_{1,2}\subset
\R^3$ two compact sets:
\begin{align}\label{laste5}
||\chi_{Q_1}S^{(N)}_B(z)\chi_{Q_1}||_{B_1}\leq 
\C(N) \sqrt{{\rm Vol}(Q_1)}\;\sqrt{{\rm Vol}(Q_2)}\;\langle r\rangle ^M.
\end{align}
The first and the third ones are easy consequences of \eqref{fazam6}. For the
 second and fourth ones we have to differentiate $N-1$ times in 
\eqref{fazam6} and write:
\begin{align}\label{fazam99}
S^{(N)}_B(z)&=(-1)^{N-1}(H_\infty(B)-z)^{-N}\\ 
&- B\sum_{k=0}^{N-1}
(-1)^{N-1-k}(H_\infty(B)-z)^{-N+k}T^{(k)}_B(z).\nonumber
\end{align}

\noindent{\bf Part two}.
We will now concentrate on \eqref{ympos2}, which needs a new
idea. 

\vspace{0.2cm}

{\it An heuristic argument}. First we perform some formal computations, in order to illustrate
how magnetic phases will transform the trace into a more regular
object. Assume that the operator under the trace in \eqref{laste15} 
has a jointly continuous
integral kernel; remember that the operator $S_B(z)$ defined in
\eqref{fazam3} had a magnetic phase. Commute this phase with $P_1(B)$
as in \eqref{fazam2}, and write the following formal expression for
the integral kernel of the operator whose trace we want to estimate:

\begin{align}\label{laste16}
 & 
\int_{\Gamma_\omega} dz
{f}_{FD}(z)
\left \{ \frac{}{}\int_{\Lambda_L}d\x' \right . \nonumber \\
&e^{iB\varphi_0(\x,\x')}[P_{1,\x}(0)+BA_1(\x-\x')]\tilde{g}_0(\x)G_1(\x,\x',z)g_0(\x') \nonumber
\\ 
&\cdot e^{iB\varphi_0(\x',\gamma)}P_{2,\x'}(0)\tilde{g}_\gamma(\x')
(H_L(0)-z-\omega )^{-1}(\x',\x'')e^{-iB\varphi_0(\x'',\gamma)}g_\gamma (\x'')\nonumber \\
&\left . + ( z \rightarrow
  z-\omega)\frac{}{} \right \}.
\end{align}
The above expression gives an integral kernel $I(\x,\x'')$. If we could prove joint continuity, 
then we could write 
\begin{align}\label{laste17}
\tilde{ X}_1(B,L,\gamma)=\int_{\Lambda_L}I(\x,\x)d\x.
\end{align}
Now let us see what happens with the phases in \eqref{laste16} when we
put $\x=\x''$. We have the identity:
\begin{align}\label{laste18}
&fl(\x,\x',\gamma):=\frac{1}{2}{\bf e}_3\cdot [(\x' -\x)\wedge
(\gamma-\x')], \nonumber \\
& \varphi_0(\x,\x')+\varphi_0(\x',\gamma)=\varphi_0(\x,\gamma)+fl(\x,\x',\gamma).
\end{align}
The crucial thing is that when $\x=\x''$ in \eqref{laste16}, the
magnetic phases cancel each other and only the flux $fl$ remains. 

Thus the operator given by the following integral kernel
\begin{align}\label{laste19}
 & 
\int_{\Gamma_\omega} dz
{f}_{FD}(z)
\left \{ \frac{}{}\int_{\Lambda_L}d\x' \right . 
e^{iB fl(\x,\x',\gamma )}[P_{1,\x}(0)+BA_1(\x-\x')]\tilde{g}_0(\x)G_1(\x,\x',z)g_0(\x') \nonumber
\\ 
&\cdot P_{2,\x'}(0)\tilde{g}_\gamma(\x')
(H_L(0)-z-\omega )^{-1}(\x',\x'')g_\gamma (\x'')\left . + ( z \rightarrow
  z-\omega)\frac{}{} \right \},
\end{align}
must have the same trace since its kernel has the same diagonal value. Remember that this is just an heuristic 
argument, precise details are given in the next paragraph.   

\vspace{0.2cm}

{\it The rigorous argument}. 
We now start the rigorous proof of \eqref{ympos2}. We integrate by parts with respect to $z$ in 
\eqref{laste15}, and let us first focus on one term, namely the one obtained when all derivatives act on the
resolvent in the middle:
\begin{align}\label{laste25}
 & R(B,L,\gamma):={\rm Tr}
\int_{\Gamma_\omega}
\tilde{f}_{FD}(z) 
\left \{ \frac{}{}P_1(B)\tilde{g}_0S_B(z)g_0 \right .
 e^{iB\varphi_0(\cdot,\gamma)}P_{2}(0)\tilde{g}_\gamma
(H_L(0)-z-\omega )^{-N}e^{-iB\varphi_0(\cdot,\gamma)}g_\gamma \nonumber \\
&\left . + ( z \rightarrow
  z-\omega)\frac{}{} \right \} dz, \quad N\geq 3.
\end{align}

Denote by $\{\psi_k\}_{k\geq 1}$ and $\{\lambda_k\}_{k\geq 1}$ the
eigenfunctions and eigenvalues of $H_L(0)$ respectively. Then the operator: 
 \begin{align}\label{laste26}
 & R_K:=\sum_{j=1}^K
\int_{\Gamma_\omega}
\tilde{f}_{FD}(z) 
\left \{ \frac{}{}P_1(B)\tilde{g}_0S_B(z)g_0 \right . 
e^{iB\varphi_0(\cdot,\gamma)}P_{2}(0)\tilde{g}_\gamma
|\psi_k\rangle \langle \psi_k|\frac{1}{(\lambda_k-z-\omega )^{N}}e^{-iB\varphi_0(\cdot,\gamma)}g_\gamma \nonumber \\
&\left . + ( z \rightarrow
  z-\omega)\frac{}{} \right \} dz
\end{align}
is trace class. Using 
\begin{align}\label{laste28}
&{\rm Tr}(R_K)-R(B,L,\gamma) {\rm Tr}
\int_{\Gamma_\omega}
\tilde{f}_{FD}(z) 
\left \{ \frac{}{}P_1(B)\tilde{g}_0S_B(z)g_0 \right
.
e^{iB\varphi_0(\cdot,\gamma)}P_{2}(0)\tilde{g}_\gamma(H_L(0)-z-\omega)^{-1}\nonumber \\
&\cdot \sum_{j>K}|\psi_k\rangle \langle \psi_k|\frac{1}{(\lambda_k-z-\omega )^{N-1}}e^{-iB\varphi_0(\cdot,\gamma)}g_\gamma \left . + ( z \rightarrow
  z-\omega)\frac{}{} \right \} dz  
\end{align}
we obtain
\begin{align}\label{laste27}
 \lim_{K\to\infty}{\rm Tr}(R_K)=R(B,L,\gamma),
\end{align}
where we use the fact that the square of the resolvent is trace class (here $N-1\geq 2$), together
with the boundedness of the rest of the factors, with norms polynomially bounded in $\langle r\rangle$.

Now let us show that $R_K$ has an integral kernel $R_K(\x,\x')$ which is jointly
continuous on $\Lambda_L\times \Lambda_L$. For, choose a
point $(\x_0,\x_0')\in \Lambda_L\times \Lambda_L$ and let us prove continuity there. We know by elliptic
regularity that $\psi_k$'s are smooth in $\Lambda_L$. Then the
function  $P_{2}(0)\tilde{g}_\gamma \psi_k$ is smooth in $\Lambda_L$. Denote by $\chi_0$ a 
smoothed-out characteristic
function of a small ball around $\x_0$, whose support is included in
$\Lambda_L$. The operator $\tilde{g}_0S_B(z)g_0$ sends
smooth functions into smooth functions, hence
$\tilde{g}_0S_B(z)g_0\chi_0P_{2}(0)\tilde{g}_\gamma \psi_k$ is smooth
in $\Lambda_L$. Then since the integral kernel of $S_B$ is smooth
outside its diagonal, it means that
$\tilde{g}_0S_B(z)g_0(1-\chi_0)P_{2}(0)\tilde{g}_\gamma \psi_k$ is
smooth at $\x_0$, and remains so even after applying $P_1(B)$. All
bounds are growing at most like a polynomial in $\langle r\rangle$,
hence the integral in $z$ preserves the continuity. The variable
$\x_0'$ only meets smooth functions in $\Lambda_L$, and we are done.

We thus conclude that $R_K$ is trace class and has a continuous integral
kernel. For every increasing sequence of compacts $\Omega_s$ such that
$\Omega_s \nearrow \Lambda_L$ (in the sense that $\Omega_s \subset
\Lambda_L$ and $\lim_{s\to\infty}{\rm
  Vol}(\Lambda_L\setminus \Omega_s)= 0$), we can 
write:
\begin{align}\label{laste29}
 & {\rm Tr} (R_K)=\lim_{s\to\infty} \int_{\Omega_s} R_K(\x,\x)d\x.
\end{align}
Let us denote by $Q_\gamma (B,z)$ the operator given by the integral kernel
\begin{align}\label{laste20}
&Q_\gamma (\x,\x';B,z)\\
&:=e^{iB fl(\x,\x',\gamma )}
[P_{1,\x}(0)+BA_1(\x-\x')]\tilde{g}_0(\x)G_1(\x,\x',z)g_0(\x')\tilde{\tilde{g}}_\gamma(\x'),\nonumber 
\end{align}
At this point we can get rid of the magnetic phases using 
\eqref{laste18}, and using the notation from \eqref{laste20} we have 
\begin{align}\label{laste30}
 &{\rm Tr}( R_K)=\sum_{j=1}^K
{\rm Tr}\int_{\Gamma_\omega}
\tilde{f}_{FD}(z) 
\left \{ \frac{}{}Q_\gamma(B,z) \right . P_{2}(0)\tilde{g}_\gamma
|\psi_k\rangle \langle \psi_k|\frac{1}{(\lambda_k-z-\omega )^{N}}
g_\gamma \nonumber \\
&\left . + ( z \rightarrow
  z-\omega)\frac{}{} \right \} dz.
\end{align}
Now $Q_\gamma$ is a nice bounded operator, and we can let $K\to
\infty$, because the integral will converge in the trace class in the
same way as in \eqref{laste28}. Then \eqref{laste27} implies that:
\begin{align}\label{laste31}
 &R(B,L,\gamma){\rm Tr}\int_{\Gamma_\omega}
\tilde{f}_{FD}(z) 
\left \{ \frac{}{}Q_\gamma(B,z) \right . P_{2}(0)\tilde{g}_\gamma
(H_L(0)-z-\omega )^{-N}
g_\gamma \nonumber \\
&\left . + ( z \rightarrow
  z-\omega)\frac{}{} \right \} dz.
\end{align}

Denote by $\tilde{Q}_\gamma(B,z)$ the
operator given by the integral kernel
\begin{align}\label{laste23}
&\tilde{Q}_\gamma (\x,\x';B,z):=\frac{1}{B}(e^{iB fl(\x,\x',\gamma )}-1)
 P_{1,\x}(0)\tilde{g}_0(\x)G_1(\x,\x',z)g_0(\x')\tilde{\tilde{g}}_\gamma(\x')\nonumber
\\
&+ e^{iB fl(\x,\x',\gamma )}
A_1(\x-\x')\tilde{g}_0(\x)G_1(\x,\x',z)g_0(\x')\tilde{\tilde{g}}_\gamma(\x'). 
\end{align}
We can write:
\begin{align}\label{laste32}
 & R(B,L,\gamma)-R(0,L,\gamma)=B\;{\rm Tr}
\int_{\Gamma_\omega}
\tilde{f}_{FD}(z) 
\left \{ \frac{}{}   \tilde{Q}_\gamma          P_{2}(0)\tilde{g}_\gamma
  (H_L(0)-z-\omega )^{-N}g_\gamma \right .\nonumber \\
&\left . + ( z \rightarrow
  z-\omega)\frac{}{} \right \}.
\end{align}
We also note the estimates:
\begin{align}\label{laste22}
|fl(\x,\x',\gamma )|\leq \frac{1}{2}|\x-\x'|\; |\x'-\gamma|,\quad
|e^{iB fl(\x,\x',\gamma )}-1|\leq B|fl(\x,\x',\gamma )|.
\end{align}
Now it is easy to see from \eqref{laste23}, \eqref{laste22} and
\eqref{intkernrez} that $\tilde{Q}_\gamma$ belongs to
$B(L^2)$, with a norm bounded by $\langle r\rangle^M L^\alpha$. 
By writing 
\begin{align}\label{laste34}
(H_L(0)-z-\omega)^{-2}g_\gamma
&=(H_L(0)-z-\omega)^{-1}\tilde{g}_\gamma(H_L(0)-z-\omega)^{-1}g_\gamma\\
&+(H_L(0)-z-\omega)^{-1}(1-\tilde{g}_\gamma)(H_L(0)-z-\omega)^{-1}
g_\gamma,\nonumber 
\end{align}
we can see that the operator $(H_L(0)-z-\omega
)^{-2}g_\gamma$ is trace class and  
\begin{align}\label{laste33}
||(H_L(0)-z-\omega)^{-2}g_\gamma ||_{B_1}\leq \C \cdot \langle
  r\rangle ^M L^{3\alpha},
\end{align}
after estimating the Hilbert-Schmidt norm of each factor in the two
terms. The second one will be exponentially small due to the support
properties of $g_\gamma$'s. 

We have thus proved 
\begin{align}\label{ympos22}
 \sup_{\gamma\in E}\sup_{0<|B|\leq 1}\left \vert
  \frac{1}{B}\{R(B,L,\gamma)-R(0,L,\gamma)\}\right \vert \leq
\C\cdot 
L^{4\alpha}.
\end{align}
After summation over $\gamma$, the bound is like $L^{2+2\alpha}$, and if $\alpha<1/2$ it will not contribute to 
the thermodynamic limit. 
 
Remember that this was just one possible term arising after integrating by parts $N-1$ times in 
\eqref{laste15}. All other terms having 
sufficiently many derivatives acting on the resolvent, can be treated in a similar way. A different 
class of terms is represented by the one in which all derivatives act on
$S_B(z)$.  Let us define:
\begin{align}\label{laste40}
 &R_1(B,L,\gamma)\nonumber \\
&:= {\rm Tr}
\int_{\Gamma_\omega}
\tilde{f}_{FD}(z) 
\left \{ \frac{}{}P_1(B)\tilde{g}_0S^{(N)}_B(z)g_0 \right
. 
e^{iB\varphi_0(\cdot,\gamma)}P_{2}(0)\tilde{g}_\gamma(H_L(0)-z-\omega)^{-1}
e^{-iB\varphi_0(\cdot,\gamma)}g_\gamma \nonumber \\
&\left . + ( z \rightarrow
  z-\omega)\frac{}{} \right \} dz.
\end{align}
We commute back $P_2(0)$ over the phase at its left and write:

\begin{align}\label{laste41}
&Q_2(B,z,\gamma):=P_1(B)\tilde{g}_0S^{(N)}_B(z)g_0
 P_2(B)\tilde{\tilde{g}}_\gamma, \\
 &R_1(B,L,\gamma)= {\rm Tr}
\int_{\Gamma_\omega}\tilde{f}_{FD}(z)
\left \{ \frac{}{}Q_2(B,z,\gamma)\right .
e^{iB\varphi_0(\cdot,\gamma)}\tilde{g}_\gamma(H_L(0)-z-\omega)^{-1}
e^{-iB\varphi_0(\cdot,\gamma)}g_\gamma \nonumber \\
&\left . + ( z \rightarrow
  z-\omega)\frac{}{} \right \} dz.\nonumber 
\end{align}
Now if $N$ is large enough, by using \eqref{fazam2}, the regularity of
$G_N(\x,\x';z)$, and the exponential decay of the kernels, one can
prove an estimate ($N$ large enough):
\begin{align}\label{laste42}
&|Q_2(B,z,\gamma)(\x,\x')|\leq \C(N) \cdot \langle r\rangle ^M
e^{-\frac{\delta}{\langle r\rangle}|\x-\x'|}\tilde{\tilde{g}}_\gamma(\x').
\end{align}
It means that the integrand in $R_1$ is a product of two
Hilbert-Schmidt operators. We can again introduce the cut-off with the
spectral projection of $H_L(0)$, get rid of the magnetic phases and
introduce the more regular phases, and so on. We consider that Lemma
\ref{lema30} is proved.\qed

\vspace{0.5cm}

Besides $X_1$ treated in the previous lemma, there is only one other
boundary term which needs special attention. This term is the one containing a double boundary sum:
\begin{align}\label{laste43}
 &X_{2}(B,L)= \sum_{\gamma,\gamma'\in
    E} X_2(B,L,\gamma,\gamma'),\\
&\label{laste44}
 X_2(B,L,\gamma,\gamma'):={\rm Tr}
\int_{\Gamma_\omega}
{f}_{FD}(z) \nonumber \\
 &\cdot 
\left \{\frac{}{}e^{iB\varphi_0(\cdot,\gamma')}P_{1,\gamma'}(B)\tilde{g}_{\gamma'}
(H_L(B,\gamma')-z )^{-1}e^{-iB\varphi_0(\cdot,\gamma')}g_{\gamma'}\right . \nonumber
\\ 
&\cdot e^{iB\varphi_0(\cdot,\gamma)}P_{2,\gamma}(B)\tilde{g}_\gamma
(H_L(B,\gamma)-z-\omega )^{-1}e^{-iB\varphi_0(\cdot,\gamma)}g_\gamma \nonumber \\
&\left . + ( z \rightarrow
  z-\omega)\frac{}{} \right \} dz.
\end{align}
and we want to prove that for every $L\geq 1$  we have 
\begin{align}\label{laste46}
\sup_{\gamma,\gamma'\in E}\sup_{0<|B|\leq 1}\left \vert
  \frac{1}{B}\{X_2(B,L,\gamma,\gamma')-X_2(0,L,\gamma,\gamma')\}\right \vert \leq
\C(\alpha)\cdot 
L^{5\alpha} .
\end{align}
Note that this would again imply something like \eqref{laste13} but
for $X_2$, because there is a finite, $L$-independent number of $\gamma$'s and
$\gamma'$'s with joint support.

The strategy is the same. We define
 \begin{align}\label{laste45}
&\tilde{ X}_2(B,L,\gamma,\gamma'):={\rm Tr}
\int_{\Gamma_\omega} f_{FD}(z) \nonumber \\
 &\cdot \left \{ \frac{}{} e^{iB\varphi_0(\cdot,\gamma')}P_{1}(0)\tilde{g}_{\gamma'}
(H_L(0)-z )^{-1}e^{-iB\varphi_0(\cdot,\gamma')}g_{\gamma'} \right .\nonumber
\\ 
&\cdot e^{iB\varphi_0(\cdot,\gamma)}P_{2}(0)\tilde{g}_\gamma
(H_L(0)-z-\omega )^{-1}e^{-iB\varphi_0(\cdot,\gamma)}g_\gamma \nonumber \\
& + \left . ( z \rightarrow
  z-\omega)\frac{}{}  \right \} dz,
\end{align}
and we want to prove two analogues of \eqref{ympos1} and
\eqref{ympos2}. The analogue of \eqref{ympos1} is ``easy'', but the
analogue of \eqref{ympos2} again requires a limiting procedure which
would allow us to write the $\tilde{ X}_2(B,L,\gamma,\gamma')$ as the
trace of a more regular object in $B$. The main point is that this
new object will be given by the composition of those four magnetic
phases present in $\tilde{ X}_2$. Namely, let us notice the following
identity:
\begin{align}\label{laste47}
&\varphi_0(\x,\gamma')+\varphi_0(\gamma',\x')+\varphi_0(\x',\gamma)+\varphi_0(\gamma,\x)\\
&=fl(\x,\gamma',\x')+fl(\x',\gamma,\x)\nonumber \\
&=fl(\x,\gamma',\gamma)+fl(\x',\gamma,\gamma').\label{laste48}
\end{align}
Now \eqref{laste48} allows us to write:
\begin{align}\label{laste49}
&\tilde{ X}_2(B,L,\gamma,\gamma')={\rm Tr}
\int_{\Gamma_\omega}
{f}_{FD}(z) 
\left \{ \frac{}{}e^{iBfl(\cdot ,\gamma',\gamma)}P_{1}(0)\tilde{g}_{\gamma'}
(H_L(0)-z )^{-1}g_{\gamma'}\right . \nonumber
\\ 
&\cdot e^{iBfl(\cdot ,\gamma,\gamma')}P_{2}(0)\tilde{g}_\gamma
(H_L(0)-z-\omega )^{-1}g_\gamma \nonumber \\
&\left . + ( z \rightarrow
  z-\omega)\frac{}{} \right \} dz.
\end{align}
The good thing about this formula is that on the supports of
$g_\gamma$'s, these fluxes are at most of order $L^{2\alpha}$, being
bounded from above by $|\x-\gamma|\cdot |\gamma-\gamma'|$. Remember
that the non-zero terms must have $|\gamma-\gamma'|\leq \C\cdot
L^\alpha$. Now we can expand the exponentials and prove the analogue
of \eqref{ympos2}. Thus Proposition \ref{lemma300} is proved.\qed 
\vspace{0.5cm}

\subsubsection{The bulk contribution}

At this point we are left with the contribution coming from terms
only containing $H_\infty(B)$. Define: 
\begin{align}\label{finish1}
& X_0(B,L):= {\rm Tr}
\int_{\Gamma_\omega}
{f}_{FD}(z)\left \{\frac{}{} P_1(B)\tilde{g}_0(H_\infty(B)-z)^{-1}g_0 P_2(B) 
\tilde{g}_0(H_\infty(B)-z-\omega)^{-1}g_0 \right . \nonumber \\
&+  \left . z \rightarrow z-\omega  \frac{}{} \right \}dz.
\end{align}
We will compute $\frac{1}{{\rm Vol}(\Lambda_L)}\partial_BX_0(0,L)$ and
show that it converges to $\{\partial_B\sigma_\infty\}(0)$.  

Define:
\begin{align}\label{finish2}
&\tilde{ X}_0(B,L):= {\rm Tr}
\int_{\Gamma_\omega}
{f}_{FD}(z) \left \{\frac{}{} P_1(B)\tilde{g}_0S_B(z)g_0 P_2(B) 
\tilde{g}_0S_B(z+\omega)g_0 \right .+  \left . z \rightarrow z-\omega  \frac{}{} \right \}dz.
\end{align}

Now we can prove the last technical result:
\begin{proposition}\label{kass}
The following two double limits exist:
\begin{align}\label{finish3}
\sigma_1 &:=\lim_{L\to\infty}\frac{1}{{\rm Vol}(\Lambda_L)}\lim_{B\to 0}
\frac{1}{B}\{X_0(B,L)-\tilde{X}_0(B,L)\},\\
\label{finish4}\sigma_2&:=\lim_{L\to\infty}\frac{1}{{\rm Vol}(\Lambda_L)}
\lim_{B\to 0}\frac{1}{B}\{\tilde{X}_0(B,L)-{X}_0(0,L)\}.
\end{align}
Moreover, the mapping $s_B$ defined in Theorem \ref{teorema1} is
differentiable at $B=0$ and
\begin{align}\label{finish5}
\lim_{L\to\infty}\{\partial_B\sigma_{L}\}(0)&\lim_{L\to\infty}\frac{1}{{\rm
    Vol}(\Lambda_L)}\{\partial_BX_0\}(0,L)\nonumber \\
&=\sigma_1+\sigma_2=\{\partial_B\sigma_\infty\}(0)=-\int_\Omega
\{\partial_B s_{B}\}_{B=0}(\x)d\x.
\end{align}
\end{proposition}

\noindent{\bf Proof}. Let us start with \eqref{finish3}. Using \eqref{fazam6} 
one can show the following identity:
\begin{align}\label{finish3a}
&\lim_{B\to 0}\frac{1}{B}\{X_0(B,L)-\tilde{X}_0(B,L)\}=-{\rm Tr}
\int_{\Gamma_\omega}
{f}_{FD}(z) \nonumber \\
&\cdot\left \{\frac{}{} P_1(0)\tilde{g}_0(H_\infty(0)-z)^{-1}T_0(z)g_0 P_2(0) 
\tilde{g}_0(H_\infty(0)-z-\omega)^{-1}g_0 \right . \nonumber \\
&+  \left . z \rightarrow z-\omega  \frac{}{} \right \}dz-{\rm Tr}
\int_{\Gamma_\omega}{f}_{FD}(z) \nonumber \\
&\cdot\left \{\frac{}{} P_1(0)\tilde{g}_0(H_\infty(0)-z)^{-1}g_0 P_2(0) 
\tilde{g}_0(H_\infty(0)-z-\omega)^{-1}T_0(z+\omega)g_0 \right . \nonumber \\
&+  \left . z \rightarrow z-\omega  \frac{}{} \right \}dz.
\end{align}
Then by integrating many times by parts, the integrand will become trace 
class, and we can get rid of the cut-off functions $\tilde{g}_0$ and $g_0$ 
since their removal will only contribute with a surface correction. 
Hence we can write:
\begin{align}\label{finish3b}
&\sigma_1=-\lim_{L\to \infty}\frac{1}{{\rm Vol(\Lambda_L)}}{\rm Tr}\;
\chi_{\Lambda_L} 
\int_{\Gamma_\omega}
{f}_{FD}(z) \nonumber \\
&\cdot\left \{\frac{}{}
P_1(0)(H_\infty(0)-z)^{-1}T_0(z) P_2(0) 
(H_\infty(0)-z-\omega)^{-1} \right . \nonumber \\
&+  \left . z \rightarrow z-\omega  \frac{}{} \right \}dz-
\lim_{L\to \infty}\frac{1}{{\rm Vol(\Lambda_L)}}{\rm Tr}\;\chi_{\Lambda_L} 
\int_{\Gamma_\omega}{f}_{FD}(z) \nonumber \\
&\cdot\left \{\frac{}{}
 P_1(0)(H_\infty(0)-z)^{-1} P_2(0) 
(H_\infty(0)-z-\omega)^{-1}T_0(z+\omega) \right . \nonumber \\
&+  \left . z \rightarrow z-\omega  \frac{}{} \right \}dz.
\end{align}
Now one can prove (as we did for $F_\infty$) that the two operators defined 
above by integrals over $\Gamma_\omega$ have jointly continuous integral 
kernels, whose diagonal values are $\mathbb{Z}^3$-periodic. 
It means that the limit 
exists and equals the integral of the kernels' diagonal value over the unit 
cube in $\R^3$. 

Now let us continue with the proof of \eqref{finish4}. First, we can get rid 
of $\tilde{g}_0$ because $P_j(B)$ is local. Let us integrate by parts $N$ 
times with respect to $z$, $N$ large. Then a typical term in the integrand 
defining $\tilde{X}_0(B,L)$ will be:
$$P_1(B)S_B^{(N+1-k)}(z)g_0 P_2(B)S_B^{(k+1)}
(z+\omega)g_0,\quad k\in\{0,...,N\},$$
where as before $S_B^{(N)}(z)$ has the integral kernel 
$e^{iB\phi_0(\x,\x')}G_{N}(\x,\x';z)$.

This operator will have an integral kernel given by:
\begin{align}\label{d3322}
&\int_{\R^3}P_{1,\x}(B)e^{iB\varphi_0(\x,\y)}G_{N+1-k}(\x,y;z)g_0(\y)\nonumber \\
&\cdot P_{2,\y}(B)e^{iB\varphi_0(\y,\x')}G_{k+1}(\y,\x';z+\omega)g_0(\x')d\y.
\end{align}
We commute the momenta with the magnetic phases and obtain:
\begin{align}\label{d3323}
&\int_{\R^3}e^{iB\varphi_0(\x,\y)}\{P_{1,\x}(0)+BA_1(\x-\y)\}
G_{N+1-k}(\x,\y;z)g_0(\y)\nonumber \\
&e^{iB\varphi_0(\y,\x')}\{P_{2,\y}(0)+BA_2(\y-\x')\}G_{k+1}(\y,\x';z+\omega)
g_0(\x')d\y.
\end{align}
This integral is absolutely convergent and defines a continuous function in 
$\x$ and $\x'$ (we can see this from the regularity and exponential 
localization of $G_{N}(\x,\x';z)$ and its first order derivatives). In order 
to  
perform the trace of this operator we put $\x=\x'$. 
The two magnetic phases will disappear, 
thus we get:
\begin{align}\label{d3324}
&\int_{\R^3}\{P_{1,\x}(0)+BA_1(\x-\y)\}
G_{N+1-k}(\x,\y;z)g_0(\y)\nonumber \\
&\{P_{2,\y}(0)+BA_2(\y-\x)\}G_{k+1}(\y,\x;z+\omega)
g_0(\x)d\y.
\end{align}
The contribution to $\lim_{B\to 0}\frac{1}{B}\{\tilde{X}_0(B,L)-{X}_0(0,L)\}$ 
coming from this term will be:
\begin{align}\label{d3325}
&R_L(\x)\\
&:=g_0(\x)\int_{\R^3}A_1(\x-\y)
G_{N+1-k}(\x,\y;z)g_0(\y)A_2(\y-\x)G_{k+1}(\y,\x;z+\omega)
d\y.\nonumber 
\end{align}
Now we have to investigate the existence of the limit:
\begin{align}\label{dpdp}
\lim_{L\to\infty}\frac{1}{{\rm
    Vol}(\Lambda_L)}\int_{\Lambda_L}R_L(\x)d\x.
\end{align}
 Let us first note that due to the 
exponential localization of $G_k$'s (see \eqref{intkernrez}) we have the 
following uniform estimate:
\begin{align}\label{d3326}
&\sup_{\x\in \R^3}\int_{\R^3}|A_1(\x-\y)|\; 
|G_{N+1-k}(\x,\y;z)|\; |A_2(\y-\x)|\;|G_{k+1}(\y,\x;z+\omega)|
d\y\nonumber \\ 
&\leq \C\cdot r^M.
\end{align}
If we look back at the definition of $g_0$, we see that it equals $1$ on the 
complementary in $\Lambda_L$ 
of a boundary neighborhood like $\Xi_L(t_0)$ with $t_0>1$ (see 
\eqref{margine1}). Denote by $\chi_L$ the characteristic function of 
$\Lambda_L\setminus \Xi_L(2t_0)$. Thus we have 
\begin{align}\label{d3327}
\chi_L\; g_0=\chi_L,\quad {\rm dist}\{{\rm supp}(1-g_0),{\rm supp}(\chi_L)\}
\geq t_0L^\alpha.
\end{align}
Because of the uniform estimate \eqref{d3326}, 
the limit in \eqref{dpdp} exists if and only if the following one exists:
\begin{align}\label{dpdp2}
\lim_{L\to\infty}\frac{1}{{\rm
    Vol}(\Lambda_L)}\int_{\Lambda_L}\chi_L(\x)R_L(\x)d\x,
\end{align}
because the difference between integrands only gives a surface contribution. 
Let us now define:
\begin{align}\label{dpdp3}
&R_\infty(\x)\\
&:=\int_{\R^3}A_1(\x-\y)
G_{N+1-k}(\x,\y;z)A_2(\y-\x)G_{k+1}(\y,\x;z+\omega)
d\y.\nonumber 
\end{align}
The difference between $\chi_L R_L$ and $\chi_L R_\infty$ comes from the 
integration over the support of $1-g_0$. But due to \eqref{d3327} and 
the exponential decay of $G_k$'s, this difference is of order 
$e^{-\delta L^\alpha/\langle r\rangle}$, 
thus will not contribute to the limit. Moreover, $R_\infty$ is 
$\mathbb{Z}^3$-periodic, therefore we can write:
\begin{align}\label{dpdp44}
\lim_{L\to \infty}\sup_{z\in \Gamma_\omega}\langle r\rangle^{-M}\left \vert \frac{1}{{\rm
    Vol}(\Lambda_L)}\int_{\Lambda_L}R_L(\x)d\x-\int_\Omega R_\infty(\x)d\x
\right\vert =0,
\end{align}
where $M$ is some large enough positive number. Then the exponential decay 
of $f_{FD}$ will insure the convergence of the $\Gamma_\omega$-integrals, thus 
\eqref{finish4} is proved. 

The last ingredient in the proof of Proposition \ref{kass} is the computation 
of $\partial_B\sigma_\infty(0)$ and the comparison with $\sigma_1+\sigma_2$. 
But the steps are very similar to those we have already done in order to 
compute $\sigma_1$ and $\sigma_2$. First, one integrates by parts many times 
with respect to $z$ in order to obtain a ``nice'' form for $F_\infty$. Second, 
using the magnetic perturbation theory one writes down a Taylor expansion 
in $B$ of $s_B(\x)$ at $B=0$ which only contains ``regularized'' terms {\it and where we can interchange the 
expansion in $B$ with the thermodynamic limit $L\to\infty$}. This strategy has 
been already used in \cite{CNP} for the Faraday effect (including the spin 
contribution, neglected here), and in \cite{BCL2} for generalized susceptibilities.  \qed

\vspace{0.5cm}

\section{Appendix:  Uniform exponential decay}

The following proposition contains two key estimates which we are going to use throughout this paper. 
\begin{proposition}\label{prop1}
Assume that $z\in \mathbb{C}$ and ${\rm dist}\{z,[0,\infty)\}= \eta>0 $. 
Then for any
$\alpha\in\{1,2,3\}$ we have 
\begin{align}\label{keyest}
&\sup_{L>1}||[D_\alpha +Ba_\alpha]\left \{(-i\nabla +B {\bf a})_D^2-z\right
\}^{-1}||\nonumber \\
&\leq \sqrt{1/\eta+\max\{\Re(z),0\}/\eta^2}.
\end{align}
Moreover, there exists a constant $C$ such that: 
\begin{align}\label{keyest2}
\sup_{L>1}\sup_{\Re(z)\geq 0}\langle |z|\rangle 
^{-1}||[D_\alpha +Ba_\alpha](H_L(B)-z)^{-1}||\leq C/\eta.
\end{align}
\end{proposition}
\noindent{\bf Proof.} The estimate \eqref{keyest} is an easy consequence of the
following trivial identity, valid for every $\psi\in L^2(\Lambda_L)$:
\begin{align}\label{keyest3}
&\sum_{\alpha=1}^3||[D_\alpha +Ba_\alpha]\left \{(-i\nabla +B {\bf a})_D^2-z\right
\}^{-1}\psi||^2 \\
&=\Re \left (\langle  \{(-i\nabla +B {\bf a})_D^2-z
\}^{-1}\psi,\psi \rangle \right )+\Re(z)||\{(-i\nabla +B {\bf a})_D^2-z
\}^{-1}\psi||^2.\nonumber
\end{align}
The estimate \eqref{keyest2} is a bit more involved. 
From \eqref{keyest} we have
that for every $\lambda>1$:

$$\sup_{L>1}||[D_\alpha +Ba_\alpha]\left \{(-i\nabla +B {\bf a})_D^2-i\lambda \right
\}^{-1}||\leq \frac{C}{\sqrt{\lambda}}.$$
Since $V$ is bounded we have: 
$$\sup_{L>1}||V[(-i\nabla +B {\bf a})_D^2-i\lambda]^{-1}||\leq \frac{C}{\lambda}.$$
Choosing a $\lambda_0$ large enough and using the Neumann series in
$V$ for
the resolvent we have 
$$\sup_{L>1}\left\{||V(H_L(B)-i\lambda_0)^{-1}||+|| [(-i\nabla +B
  {\bf a})_D^2-i\lambda_0](H_L(B)-i\lambda_0)^{-1}||\right \}\leq
1/2.$$
Using the resolvent identity we obtain:
\begin{equation}\label{222}
\sup_{L>1}||[(-i\nabla +B
  {\bf a})_D^2-i\lambda_0](H_L(B)-z)^{-1}||\leq C
  |z|/\eta.
\end{equation}
Hence writing 
\begin{align}
[D_\alpha +Ba_\alpha](H_L(B)-z)^{-1}&=[D_\alpha +Ba_\alpha][(-i\nabla +B
  {\bf a})_D^2-i\lambda_0]^{-1}\nonumber \\
&\cdot [(-i\nabla +B
  {\bf a})_D^2-i\lambda_0](H_L(B)-z)^{-1}\nonumber 
\end{align}
we obtain the result.
\qed 

\vspace{0.5cm}

We will need a certain type of uniform
exponential localization, stated in the proposition below. If $\x_0$ is some point in $\Lambda_L$ and $\alpha\in \R$, then let $e^{\alpha \langle \cdot -\x_0\rangle }$ denote the multiplication operator with the function $e^{\alpha \sqrt{ |\x -\x_0|^2+1}}$. Note that multiplication with the exponential weight is a bounded operator if $L<\infty$ and leaves invariant the domain of $H_L(B)$. 

\begin{proposition}\label{prop2}
Fix $\x_0\in\Lambda_L$ and ${\rm
  dist}\{z,[0,\infty)\}= \eta>0 $. Denote by $r:=\Re(z)$. Then there exists a
$\delta_0>0$ and a constant $C$ such that for every $0\leq \delta\leq \delta_0$ we have
 \begin{align}\label{expdek}
\sup_{L>1}\sup_{r\in \R}\sup_{\x_0\in\Lambda}\left \Vert 
e^{{\langle \cdot -\x_0\rangle \frac{\pm\delta}{\langle r\rangle }}}
(H_L(B)-z)^{-1}e^{\langle \cdot -\x_0\rangle {\frac{\mp\delta}
{\langle r\rangle}}}\right
\Vert \leq C,
\end{align}
\begin{align}\label{expdek2}
&\sup_{L>1}\sup_{r\in \R}\sup_{\x_0\in\Lambda}\left \{
  \langle r\rangle ^{-1}\left \Vert [D_\alpha+Ba_\alpha]
e^{{\langle\cdot -\x_0\rangle\frac{\pm\delta}{\langle r\rangle 
}}}(H_L(B)-z)^{-1}e^{\langle \cdot -\x_0\rangle{\frac{\mp\delta}{\langle r 
\rangle }}}\right
\Vert \right \}\leq C,
\end{align}
and 
\begin{align}\label{expdek3}
&\sup_{L>1}\sup_{r\in \R}\sup_{\x_0\in\Lambda}\left \{
  \langle r\rangle ^{-1}\left \Vert e^{{\langle \cdot -\x_0\rangle 
\frac{\pm\delta}{\langle r\rangle }}}
[(-i\nabla +B{\bf a})_D^2+1](H_L(B)-z)^{-1}e^{\langle \cdot -\x_0\rangle 
{\frac{\mp\delta}{\langle r\rangle }}}\right
\Vert \right \}\leq C. 
\end{align}
\end{proposition}

\noindent{\bf Proof.}  An heuristic explanation of \eqref{expdek} is the following: if we apply the resolvent on a function which is exponentially localized near $\x_0$ and decays like $e^{- \frac{\delta}{\langle r 
\rangle }\langle \x -\x_0\rangle }$, then we do not loose exponential decay. Moreover, the $L^2$ bounds are uniform in $z$, $L$ and the location of $\x_0$.

 that For $s\in\R$,
the well-known Combes-Thomas rotation \cite{CT} gives: 
$$e^{s\langle \cdot -\x_0\rangle }(H_L(B)-z)e^{-s\langle \cdot -\x_0\rangle 
}=H_L(B)-z+s\sum_{j=1}^3w_j[D_j+Ba_j]+sV_1+s^2V_2,$$
where $w_j, V_1,V_2$ are bounded functions, uniformly in
$L$ and $\x_0$. 

Now put $s=\delta/\langle r\rangle $, and use \eqref{keyest2}. If 
$\delta$ is small
enough, we get that uniformly in $L$, $\x_0$ and $r$ we have
$$||\{s\sum_{j=1}^3w_j[D_j+Ba_j]+sV_1+s^2V_2\}(H_L(B)-z)^{-1}||\leq
1/2,$$
which gives
\begin{align}\label{expdek4}
& e^{s\langle \cdot -\x_0\rangle}(H_L(B)-z)^{-1}e^{-s\langle \cdot -\x_0\rangle}=(H_L(B)-z)^{-1}\nonumber \\
&\cdot \left \{1+\left [s\sum_{j=1}^3w_j(D_j+Ba_j)+sV_1+s^2V_2\right
  ](H_L(B)-z)^{-1}\right \}^{-1}.
\end{align}
This implies \eqref{expdek}, and together with \eqref{keyest2} we
also get \eqref{expdek2}. 

Let us now concentrate ourselves on the last estimate \eqref{expdek3}. Up to a
commutation, \eqref{expdek2} gives 
 \begin{align}\label{expdek5}
\sup_{L>1}\sup_{r\in \R}\sup_{\x_0\in\Lambda}\left \{
  \langle r\rangle ^{-1}\cdot \left \Vert e^{{\langle \cdot -\x_0\rangle\frac{\delta}
{\langle r\rangle }}}[D_\alpha+Ba_\alpha](H_L(B)-z)^{-1}e^{-\langle \cdot -\x_0\rangle{\frac{\delta}{\langle r\rangle }}}\right
\Vert \right \}\leq C. 
\end{align}
Thus again up to a commutation, \eqref{expdek3} follows if we
can prove
\begin{align}\label{expdek6}
\sup_{L>1}\sup_{r\in \R}\sup_{\x_0\in\Lambda}\left \{
  \langle r\rangle ^{-1}\left \Vert[(-i\nabla +B{\bf a})_D^2+1] e^{{\langle \cdot -\x_0\rangle\frac{\delta}{\langle r\rangle }}}(H_L(B)-z)^{-1}e^{-\langle \cdot -\x_0\rangle{\frac{\delta}{\langle r\rangle }}}\right
\Vert \right \}\leq C.
\end{align}
But this estimate follows from \eqref{expdek4} and \eqref{222}.

\qed

\vspace{0.5cm}

\begin{corollary}\label{cor1}
Let $\lambda\geq \lambda_0>0$. Then there exists $c>0$ such that  
\begin{align}\label{expcor}
\sup_{L>1}\sup_{\lambda\geq \lambda_0}\sup_{\x_0\in\Lambda}\left 
\Vert e^{\pm c\langle \cdot
    -\x_0\rangle 
    \sqrt{\lambda}}(H_L(B)+\lambda)^{-1}
e^{\mp c\langle \cdot
    -\x_0\rangle \sqrt{\lambda}}\right
\Vert \leq \frac{\C(\lambda_0)}{\lambda}.
\end{align}
\end{corollary}
\noindent{\bf Proof.} We use the key estimate \eqref{keyest} for the
case when $\eta\geq \lambda$ and $\Re(z)=-\lambda<0$. This gives us 
$$||(D_\alpha+B a_\alpha)[(-i\nabla+B{\bf a})_D^2+\lambda]^{-1}||\leq 
\C/\sqrt{\lambda}$$
hence 
$$||(D_\alpha+B a_\alpha)[H_L(B)+\lambda]^{-1}||\leq \C/\sqrt{\lambda}.$$
Now we proceed as in \eqref{expdek4} and we get the result. Finally,
note that by repeating the argument of Proposition \ref{prop2} we can
obtain a uniform estimate in $\lambda$, $L$ and $\x_0$:
\begin{align}\label{expcor2}
\left \Vert e^{\pm c\langle \cdot
    -\x_0\rangle 
    \sqrt{\lambda}}[(-i\nabla+B{\bf a})_D^2+\lambda](H_L(B)+\lambda)^{-1}
e^{\mp c\langle \cdot
    -\x_0\rangle \sqrt{\lambda}}\right
\Vert \leq \C.
\end{align}

\qed 

\vspace{0.5cm}

\begin{proposition}\label{prop3}
The operator $[(-i\nabla +B{\bf a})_D^2+\lambda]^{-1}$ has an integral
kernel $K_L(\x,\x')$ which is jointly continuous away from the diagonal
$\x=\x'$,  and obeys the estimate 
 \begin{align}\label{intkern}
 \vert K_L(\x,\x') \vert \leq 
\frac{e^{-\sqrt{\lambda}|\x-\x'|}}{4\pi |\x-\x'|},
\end{align}
for every $\x\neq \x'$ in $\Lambda_L$. 
\end{proposition}

\noindent {\bf Proof.} The argument is based on several well known results. 
First, one uses the
Feynman-Kac-It{\^ o} representation for the kernel of the semi-group $e^{-t
  (-i\nabla +B{\bf a})_D^2}$, $t>0$ (see \cite{BHL}) and obtains a 
diamagnetic inequality in $\Lambda_L$:
$$\left | e^{-t (-i\nabla +B{\bf a})_D^2}(\x,\x')\right |\leq  e^{t
  \Delta_D}(\x,\x')\leq (4\pi t)^{-3/2}e^{-\frac{|\x-\x'|^2}{4t}},
\quad \x,\x'\in\Lambda_L.$$
Second, we perform a Laplace transform and obtain the result.
\qed 

\vspace{0.5cm}

\begin{proposition}\label{prop4}
All the results in this section are also valid if the operators are
defined on the whole space $\R^3$ (formally $L=\infty$).
\end{proposition}

\noindent{\bf Proof.} The argument relies on various standard limiting
and cut-off arguments, which are necessary because the exponential growing factors do not invariate the operator domain of $H_\infty(B)$. The most important ingredient (uniformity in $L>1$ of all our previous
estimates) has been already proved. \qed

\noindent {\bf Acknowledgments.}  H.C. acknowledges support from the Danish 
F.N.U. grant {\it  Mathematical Physics}.

\end{document}